\documentclass{article}
\usepackage{latexsym, amsmath, amssymb}
\usepackage{graphicx}
\usepackage{epstopdf}
\usepackage{amsmath, amsthm, amscd, amsfonts}
\usepackage{amsmath}
\usepackage{amssymb}
\usepackage{array}
\usepackage{multirow}
\usepackage{float}
\usepackage{enumerate}
\usepackage[a4paper, total={6in, 10in}]{geometry}
\usepackage{placeins}
\usepackage{xcolor}
\usepackage{caption}
\usepackage{subcaption}

\newtheorem{thm}{Theorem}[section]

\theoremstyle{definition} \newtheorem{defn}[thm]{Definition}
\theoremstyle{question} 
\theoremstyle{remark} %\numberwithin{equation}{section} % MATH -----------------------------------------------------------

\makeatletter
\newcommand*{\rom}[1]{\expandafter\@slowromancap\romannumeral #1@}
\makeatother

\def\bege{\begin{equation}} \def\ende{\end{equation}}

   \def\begr{\begin{eqnarray}}
	\def\endr{\end{eqnarray}} 
 
\def\bege{\begin{equation}} \def\ende{\end{equation}}
\def\begr{\begin{eqnarray}} \def\endr{\end{eqnarray}} \def\bnum{\begin{enumerate}} \def\enum{\end{enumerate}}
\begin{document}
	\begin{center}\Large
		\textbf{Computational analysis of NM-polynomial based topological indices and graph-entropies of carbon nanotube Y-junctions}
	\end{center}
	\begin{center}
		Sohan Lal$^{1}$, Vijay Kumar
		Bhat$^{1,}$$^{\ast}$, Sahil Sharma$^{1}$
	\end{center}
	
	\begin{center}
		$^{1}$School of Mathematics, Shri Mata Vaishno Devi
		University,\end{center}\begin{center}Katra-$182320$, Jammu and
		Kashmir, India.
	\end{center}
	
	\begin{center}
		sohan1993sharma@gmail.com, vijaykumarbhat2000@yahoo.com, sahilsharma2634@gmail.com
	\end{center}
\begin{abstract}
	Carbon nanotube Y-junctions are of great interest to the next generation of innovative multi-terminal nanodevices. Topological indices are graph-theoretically based parameters that describe various structural properties of a chemical molecule. The entropy of a graph is a topological descriptor that serves to characterize the complexity of the underlying molecular graph. The concept of entropy is a physical property of a thermodynamic system. Graph entropies are the essential thermophysical quantities defined for various graph invariants and are applied to measure the heterogeneity and relative stabilities of molecules. In this paper, several neighborhood degree sum-based topological indices including graph-based entropies of carbon nanotube Y-junction graphs are computed.
\end{abstract}

\section{Introduction}
 Nanotechnology is currently popular because of its evolving, electron transfer property and low-cost implementation. Nanotubes \cite{53}, were discovered in $1985$ and carbon nanotubes \cite{54} in $1991$. In nanoscience and technology, branched or non-straight carbon nanotubes such as L, T, X, and Y have a lot of applications in electronic devices, such as three-terminal transistors, multi-terminal nanoelectronics, switches, amplifiers, etc., \cite{5,6,7,8,9,10}. These junctions are a great option for the production of nanoscale electronic devices with better switching and reliable transport properties at room temperature. For more applications of carbon nanotube Y-junctions, we refer to \cite{c,d,e}.\\
 
  The first proposed branched carbon nanotube was of Y shape, commonly known as Y-junction or three-terminal junction. These junctions are classified as an armchair, zig-zag, or chiral depending on the chirality of connected carbon nanotubes. Also, they can be single-walled or multi-walled, symmetric or asymmetric, capped or uncapped. A carbon nanotube is called uncapped if both ends are open. A Y-junction is called symmetric if the nanotubes joining in the Y shape are identical, heptagons appeared isolated, and are distributed symmetrically. For various symmetric and asymmetric carbon nanotube Y-junctions, we refer to \cite{19,20,21,22}.\\
  
   A carbon nanotube Y-junction is formed by joining three identical carbon nanotubes in a Y-shaped pattern. These junctions contain exactly six hexagons as well as heptagons at the branching points. The first structural model of symmetrical single-walled armchair carbon nanotube Y-junctions was proposed by Chernozatonskii \cite{13} and Scuseria \cite{12}, independently, in $1992$. These junctions were experimentally observed \cite{14} in $1995$. For more applications and properties of carbon nanotube Y-junction graphs, we refer to \cite{15,16,17}.\\
   
    Mathematical chemistry is a branch of theoretical chemistry that employs mathematical techniques to explain the molecular structure of a chemical molecule and its physicochemical properties. Molecular graphs are a visual representation of a chemical molecule with vertices representing atoms and edges representing bonds between the atoms \cite{39}. Let $G=(V(G),E(G))$ be a molecular graph with vertex set $V(G)$ and edge set $E(G)$. The order of a molecular graph $G$ is defined as the total number of vertices in $G$, denoted by $\lvert V(G)\rvert$, and the number of edges in $G$ is called size of $G$, denoted by $\lvert E(G)\rvert$. Any edge of the graph connecting its vertices $u$ and $v$, is denoted by $e=uv\in E(G)$. Two vertices of graph $G$ are said to be adjacent if there exists an edge between them. The degree of vertex $v\in V(G)$, denoted by $d(v)$,  is defined as the number of vertices that are adjacent to vertex $v$, i.e., $d(v)$= $\lvert\{u:e=uv\in E(G)\}\rvert$. The neighborhood degree sum of vertex $v\in V(G)$ is denoted by $d_n(v)$, and is defined as the sum of the degrees of all vertices that are adjacent to $v$, i.e., $d_n(v)= \sum\limits_{u}d(v)$: $uv\in E(G)$.
The minimum cardinality of the set $K\subseteq V(G)$ such that $G\setminus K$ is disconnected graph is called connectivity or vertex-connectivity of a connected graph $G$. The connected graph $G$ is said to be $k$-connected if its connectivity is $k$. \\

 Topological indices are the numerical values calculated from molecular graphs to describe various structural properties of the chemical molecule. They are frequently used to model many physicochemical properties in various quantitative structure-property/activity relationship (QSPR/QSAR) studies \cite{h,32,36}. In $1947$, the chemist Harold Wiener \cite{24} initiated the concept of topological indices. Since then, various topological indices have been introduced, and a lot of research has been conducted toward computing the indices for different molecular graphs and networks. A topological index based on the degree of end vertices of an edge can predict various physicochemical properties of the molecule, such as heat of formation, strain energy, entropy, enthalpy, boiling points, flash point, etc., without using any weight lab \cite{32}.\\
 
  The Zagreb indices and their variations have been used to investigate molecular complexity, ZE-isomerism, and chirality \cite{71}. In general, the Zagreb indices have shown applicability for deriving multilinear regression models. Ghorbani and Hosseinzadeh \cite{25} introduced the third version of the Zagreb index and shows that this index shows a good correlation with acentric factor and entropy of the octane isomers. Mondal et al. \cite{26} introduced neighborhood degree sum-based topological indices namely neighborhood version of forgotten topological index and neighborhood version of second modified Zagreb index and discuss some mathematical properties and degeneracy of these novel indices. For more neighborhood degree sum-based topological indices, their properties, and applications, we refer to \cite{32,29,37}.\\

The process of computing the topological indices of a molecular graph from their definitions is complex and time-consuming. Thus, for a particular family of graphs and networks, algebraic polynomials play an important role in reducing the computational time and complexity when computing its topological indices. In short, with the help of algebraic polynomials, one can easily compute various kinds of graph indices within a short span of time. The NM-polynomial plays vital role in the computation of neighborhood degree sum-based topological indices. Let $d_n(v)$ denotes the neighborhood degrees sum of vertex $v\in V(G)$. Then, the neighborhood M-polynomial (NM-polynomial) of $G$ is defined as \cite{29,27,28}
\begin{eqnarray}
	NM(G;x,y)=\sum\limits_{i\leq j}\lvert E_{ij}(G)\rvert x^iy^j
\end{eqnarray}
where, $\lvert E_{ij}(G)\rvert$, $i,j\geq1$, be the number of all edges $e=uv\in E(G)$ such that $\{d_n(u)=i, d_n(v)=j\}$.\\

Recently, various neighborhood degree sum-based topological indices have been computed via the NM-polynomial technique. For example, Mondal et al. \cite{29,30} obtained some neighborhood and multiplicative neighborhood degree sum-based indices of molecular graphs by using their NM-polynomials. Kirmani et al. \cite{32} and Mondal et al. \cite{31}, investigated some neighborhood degree sum-based topological indices of antiviral drugs used for the treatment of COVID-$19$ via the NM-polynomial technique. Shanmukha et al. \cite{33} computed the topological indices of porous graphene via NM-polynomial method. For more neighborhood degree sum-based topological indices via NM-polynomials, we refer to \cite{32,31,i,j}.\\

 Some neighborhood degree sum-based topological indices and their derivation from NM-polynomial are given in Table $1$.\\
 
 \begin{table}[ht]
 	\centering
 	\caption{Description of some topological indices and its derivation from NM-polynomial}
 	\setlength{\tabcolsep}{9pt}
 	\renewcommand{\arraystretch}{1.8}
 	\tiny\begin{tabular}{cccc}
 		\hline Topological index & Formula&Derivation from $NM(G;x,y)$\\\hline
 		Third version of Zagreb index \cite{25}: $NM_1(G)$ & $\sum\limits_{uv\in E(G)}\big(d_n(u)+d_n(v)\big)$&$(D_x+D_y)(NM(G;x,y))\rvert_{x=y=1}$\\\\
 		Neighborhood second Zagreb index \cite{26}: $NM_2(G)$	& $\sum\limits_{uv\in E(G)}\big(d_n(u)d_n(v)\big)$&$(D_xD_y)(NM(G;x,y))\rvert_{x=y=1}$\\\\
 		Neighborhood second modified Zagreb index \cite{29}: $^{nm}M_2(G)$ & $\sum\limits_{uv\in E(G)}\big(\frac{1}{d_n(u)d_n(v)}\big)$&$(S_xS_y)(NM(G;x,y))\rvert_{x=y=1}$\\\\ 
 		Neighborhood forgotten topological index \cite{26}: $NF(G)$ & $\sum\limits_{uv\in E(G)}\big(d_n^2(u)+d_n^2(v)\big)$&$(D^2_x+D^2_y)(NM(G;x,y))\rvert_{x=y=1}$\\\\
 		Third NDe index \cite{29}: $ND_3(G)$ & $\sum\limits_{uv\in E(G)}d_n(u)d_n(v)(d_n(u)+d_n(v))$&$D_xD_y(D_x+D_y)(NM(G;x,y))\rvert_{x=y=1}$\\\\ Neighborhood general Randic index \cite{29}: $NR_\alpha(G)$ & $\sum\limits_{uv\in E(G)}d_n^{\alpha}(u) d_n^{\alpha}(v)$&$(D^{\alpha}_xD^{\alpha}_y)(NM(G;x,y))\rvert_{x=y=1}$\\\\
 		Neighborhood inverse Randic index \cite{29}: $NRR_\alpha(G)$ & $\sum\limits_{uv\in E(G)}\frac{1}{d_n^{\alpha}(u)d_n^{\alpha}(v)}$&$(S^{\alpha}_xS^{\alpha}_y)(NM(G;x,y))\rvert_{x=y=1}$\\\\  Fifth NDe index \cite{29}: $ND_5(G)$ & $\sum\limits_{uv\in E(G)}\big(\frac{d_n^2(u)+d_n^2(v)}{d_n(u)d_n(v)}\big)$&$(D_xS_y+S_xD_y)(NM(G;x,y))\rvert_{x=y=1}$\\\\Neighborhood harmonic index \cite{29}: $NH(G)$ &$\sum\limits_{ab\in E(G)}\frac{2}{d_n(u)+d_n(v)}$&$2S_xT(NM(G;x,y))\rvert_{x=y=1}$\\\\  Neighborhood inverse sum indeg index \cite{29}: $NI(G)$ & $\sum\limits_{uv\in E(G)}\big(\frac{d_n(u)d_n(v)}{d_n(u)+d_n(v)}\big)$&$(S_xTD_xD_y)(NM(G;x,y))\rvert_{x=y=1}$\\
 		\hline
 	\end{tabular}
 	where,
 	\textbf{$D_x=x\big(\frac{(\partial( NM(G;x,y))}{\partial x}\big)$, $D_y=y\big(\frac{(\partial (NM(G;x,y))}{\partial y}\big)$, $S_x=\int_{0}^{x}\frac{NM(G;t,y)}{t}dt$, $S_y=\int_{0}^{y}\frac{NM(G;x,t)}{t}dt$,  $T(NM(G;x,y))=NM(G;x,x)$}.
 \end{table}
In chemical graph theory, the determination of the structural information content \cite{41} of a graph is mostly based on the vertex partition of a graph to obtain a probability distribution of its vertex set \cite{1}. Based on such a probability distribution, the entropy of a graph can be defined. Thus, the structural information content of a graph is defined as the entropy of the underlying graph topology. The concept of graph entropy or entropy of graph was first time appeared in \cite{38}, where molecular graphs are used to study the information content of an organism. Entropy-based methods are powerful tools to investigate various problems in cybernetics, mathematical chemistry, pattern recognition, and computational physics \cite{39,41,40,60,62}.\\

  Entropy is a measure of randomness, uncertainty, heterogeneity, or lack of information in a system. Based on information indices, there are various approaches to deriving graph entropy from the topological structure of a given chemical molecule \cite{42}. For example, Trucco \cite{41} and 
 Rashevsky \cite{38} defined graph entropies in terms of degree of vertex, extended degree sequences, and number of vertices of a molecular graph. Tan and Wu \cite{43} study network heterogeneity by using vertex-degree based entropies. Mowshowitz defined the entropy of a graph in terms of equivalence relations defined on the vertex set of a graph and
 discussed some properties related to structural information \cite{49,50,51,52}.\\

 Recently, Shabbir and Nadeem \cite{44} defined graph entropies in terms of topological indices for the molecular graphs of carbon nanotube Y-junctions and developed the regression models between the graph entropies and topological indices. Nadeem et al. \cite{45} calculated some degree-based topological indices for armchair carbon semicapped and capped nanotubes and investigated their chemical and physical properties. Bača et al. \cite{46} computed some degree-based topological indices of a carbon nanotube network and studied its properties. Azeem et al. \cite{47} calculated some M-polynomials based topological indices of carbon nanotube Y-junctions and their variants. Ahmad \cite{48}, studied some ve-degree based topological indices of carbon nanotube Y-junctions and discussed their properties. Ayesha \cite{k} calculated the bond energy of symmetrical single-walled armchair carbon nanotube Y-junctions and developed regression models between bond energy and topological indices. Rahul et al. \cite{61} calculated some degree-based topological indices and graph-entropies of graphene, graphyne, and graphdiyne by using Shannon’s approach.\\
 
 The above-mentioned literature and applications of carbon nanotubes in the field of nanoscience and technology inspired us to develop more research on the molecular structure of carbon nanotube Y-junction and their variants. In addition, no work has been reported on NM-polynomial based topological indices and index-entropies of Y-junction graphs. Therefore, the main contribution of this study includes the following:
  
 \begin{itemize}
 	\item Computation of NM-polynomials of carbon nanotube Y-junction graphs.
 	\item Computation of some neighborhood degree sum-based topological indices from NM-polynomials.
 	\item Some graph index-entropies in terms of topological indices are defined and computed.
 	\item Comparative analysis of obtained topological indices and graph index-entropies of Y-junction graphs.
 	
 \end{itemize}
\section{Aim and Methodology} 
We use the edge partition technique, graph-theoretical tools, combinatorial computation, and the degree counting method to derive our results. The degree of end vertices is used to generate the patterns of edge partitions of the Y-junction graphs. Using such partitions, a general expression of NM-polynomials is derived. Then, several neighborhood degree sum-based topological indices are obtained from the expression of these NM-polynomials with the help of Table $1$. Also, graph index- entropies in terms of topological indices have been defined by using edge-weight functions and computed for Y-junction graphs.\\

The paper is structured as follows: In Section $3$, we define topological index-based graph entropies. The Y-junction graphs and their constructions are described in Section $4$. In Section $5$, the general expression of the NM-polynomials and neighborhood degree sum-based topological indices of Y-junction graphs are presented. Section $6$ describes the graph index-entropies of Y-junction graphs. The numerical analysis of the findings is discussed in Section $7$. Finally, the conclusion is drawn and discussed in Section $8$.   
\section{Definitions and Preliminaries}
	In this section, we define graph index-entropies in terms of an edge-weight function. In $2008$, Dehmer \cite{1} defined the entropy for a connected graph $G$ as follows: 
\begin{defn} \cite{1}
	Let $G=(V(G),E(G))$ be a connected graph of order $n$ and $g$ be an arbitrary information functional. Then the entropy of $G$ is defined as
	\begin{equation}\label{1}
		H_g(G)=-\sum\limits_{i=1}^n\frac{g(v_i)}{\sum\limits_{i=1}^ng(v_i)}log\biggl(\frac{g(v_i)}{\sum\limits_{i=1}^ng(v_i)}\biggr).
	\end{equation}
\end{defn}
Since an information function defined on the vertex set of a graph is an arbitrary function. Hence, Dehmer's definition shows the possibility of producing various graph entropies for a variation in the selection of information functionals. For such graph entropy, we can refer to \cite{2,3,4}.\\

Let $\beta:E(G)\to \mathbb{R}^+\cup\{0\}$ be an edge-weight function and $d_n(u)=\sum\limits_{uv\in E(G)}d(u)$, denotes the sum of degrees of end vertices of an edges incident to vertex $u\in V(G)$ (also known as neighborhood degree-sum of vertex $u$). Then, for eight different edge-weight functions, the third-version of Zagreb index, neighborhood second Zagreb index, neighborhood forgotten topological index, neighborhood second modified Zagreb index, third NDe index, fifth NDe index, neighborhood harmonic index and neighborhood inverse sum indeg index-entropies have been defined in the following manner:
\begin{itemize}
	\item \textbf{Third-version of Zagreb index-entropy:} If $e=uv$ is an edge of a connected graph $G$ and $\beta_1(e)=d_n(u)+d_n(v)$ is an edge-weight function defined on $E(G)$. Then, the third-version of Zagreb index is
	 \begin{equation}
		 NM_1(G)=\sum\limits_{e=uv\in E(G)}\beta_1(e)=\sum\limits_{e=uv\in E(G)}d_n(u)+d_n(v).
		\end{equation}
	Equation (\ref{1}) for this edge-weight function gives us
	\begin{eqnarray*}
		H_{\beta_1}(G)&=&-\sum\limits_{e\in E(G)}\frac{\beta_1(e)}{\sum\limits_{e\in E(G)}\beta_1(e)}log\bigg(\frac{\beta_1(e)}{\sum\limits_{e\in E(G)}\beta_1(e)}\bigg)\\
		&=&-\frac{1}{\sum\limits_{e\in E(G)}\beta_1(e)}\sum\limits_{e\in E(G)}\beta_1(e)\bigg(log(\beta_1(e))-log\sum\limits_{e\in E(G)}\beta_1(e)\bigg)
	\end{eqnarray*}
\begin{eqnarray*}
		&=&-\frac{1}{\sum\limits_{e\in E(G)}\beta_1(e)}\sum\limits_{e\in E(G)}\beta_1(e)log(\beta_1(e))+\frac{1}{\sum\limits_{e\in E(G)}\beta_1(e)}\sum\limits_{e\in E(G)}\beta_1(e)log\bigg(\sum\limits_{e\in E(G)}\beta_1(e)\bigg)\\
		&=&log\bigg(\sum\limits_{e\in E(G)}\beta_1(e)\bigg)-\frac{1}{\sum\limits_{e\in E(G)}\beta_1(e)}\sum\limits_{e\in E(G)}\beta_1(e)log(\beta_1(e)).
	\end{eqnarray*}
On replacing $\sum\limits_{e\in E(G)}\beta_1(e)$ by $NM_1(G)$ in the above equation, we get the following third-version of Zagreb index-entropy
\begin{equation}
	H_{\beta_1}(G)=log(NM_1(G))-\frac{1}{NM_1(G)}\sum\limits_{e\in E(G)}{\beta}_1(e)log\beta_1(e)
.\end{equation} 
Similarly, we define other graph index-entropies as follows:
\item\textbf{Neighborhood second Zagreb index-entropy:} For $\beta_2(e)=d_n(u)d_n(v)$, the neighborhood second Zagreb index and neighborhood second Zagreb index-entropy are
\begin{equation}
	NM_2(G)=\sum\limits_{e=uv\in E(G)}d_n(u)d_n(v),
\end{equation} 
and
\begin{equation}
	H_{\beta_2}(G)=log(NM_2(G))-\frac{1}{NM_2(G)}\sum\limits_{e\in E(G)}{\beta}_2 (e)log\beta_2(e).
\end{equation}
\item\textbf{Neighborhood forgotten topological index-entropy:} For $\beta_3(e)=d_n^2(u)+d_n^2(v)$, the neighborhood forgotten topological index and neighborhood forgotten topological index-entropy are
\begin{equation}
	NF(G)=\sum\limits_{e=uv\in E(G)}d_n^2(u)+d_n^2(v),
\end{equation} 
and
\begin{equation}
	H_{\beta_3}(G)=log(NF(G))-\frac{1}{NF(G)}\sum\limits_{e\in E(G)}{\beta}_3 (e)log\beta_3(e).
\end{equation}
\item\textbf{Neighborhood second modified Zagreb index-entropy:} For $\beta_4(e)=\frac{1}{d_n(u)d_n(v)}$, the neighborhood second modified Zagreb index and neighborhood second modified Zagreb index-entropy are
\begin{equation}
	^{nm}M_2(G)=\sum\limits_{e=uv\in E(G)}\frac{1}{d_n(u)d_n(v)},
\end{equation} 
and
\begin{equation}
	H_{\beta_4}(G)=log(^{nm}M_2(G))-\frac{1}{^{nm}M_2(G)}\sum\limits_{e\in E(G)}{\beta}_4 (e)log\beta_4(e).
\end{equation}
\item\textbf{Third NDe index-entropy:} For $\beta_5(e)=d_n(u)d_n(v)\big(d_n(u)+d_n(v)\big)$, the third NDe index and third NDe index-entropy are
\begin{equation}
	ND_3(G)=\sum\limits_{e=uv\in E(G)}d_n(u)d_n(v)\big(d_n(u)+d_n(v)\big),
\end{equation} 
and
\begin{equation}
	H_{\beta_5}(G)=log(ND_3(G))-\frac{1}{ND_3(G)}\sum\limits_{e\in E(G)}{\beta}_5 (e)log\beta_5(e).
\end{equation}
\item\textbf{Fifth NDe index-entropy:} For $\beta_6(e)=\frac{d_n(u)}{d_n(v)}+\frac{d_n(v)}{d_n(u)}$, the fifth NDe index and fifth NDe index-entropy are
\begin{equation}
	ND_5(G)=\sum\limits_{e=uv\in E(G)}\frac{d_n(u)}{d_n(v)}+\frac{d_n(v)}{d_n(u)},
\end{equation} 
and
\begin{equation}
	H_{\beta_6}(G)=log(ND_5(G))-\frac{1}{ND_5(G)}\sum\limits_{e\in E(G)}{\beta}_6 (e)log\beta_6(e).
\end{equation}
\item\textbf{Neighborhood harmonic index-entropy:} For $\beta_7(e)=\frac{2}{d_n(u)+d_n(v)}$, the neighborhood harmonic index and neighborhood harmonic index-entropy are
\begin{equation}
	NH(G)=\sum\limits_{e=uv\in E(G)}\frac{2}{d_n(u)+d_n(v)},
\end{equation} 
and
\begin{equation}
	H_{\beta_7}(G)=log(NH(G))-\frac{1}{NH(G)}\sum\limits_{e\in E(G)}{\beta}_7 (e)log\beta_7(e).
\end{equation}
\item\textbf{Neighborhood inverse sum indeg index-entropy:} For $\beta_8(e)=\frac{d_n(u)d_n(v)}{d_n(u)+d_n(v)}$, the neighborhood inverse sum index and neighborhood inverse sum index-entropy are
\begin{equation}
	NI(G)=\sum\limits_{e=uv\in E(G)}\frac{d_n(u)d_n(v)}{d_n(u)+d_n(v)},
\end{equation} 
and
\begin{equation}
	H_{\beta_8}(G)=log(NI(G))-\frac{1}{NI(G)}\sum\limits_{e\in E(G)}{\beta}_8 (e)log\beta_8(e).
\end{equation}
 \end{itemize}
\section{Y-Junction Graphs}
 The Y-junctions examined in this study are created by the covalent connection of three identical single-walled carbon nanotubes crossing at an angle of $120^{\circ}$ and are uniquely determined by their chiral vector $v=nv_1+nv_2$, where $v_1$ and $v_2$ are graphene sheet lattice vectors and $n$ is non-negative integer. Let $m\geq 1$ and $n\geq 4$ be an even integer. Then, an uncapped symmetrical single-walled carbon nanotube Y-junction is made up of an armchair $Y(n,n)$ and three identical single-walled armchair carbon nanotubes $T_m(n,n)$ each of length $m$ (layers of hexogones), denoted by $Y^m(n,n)$. In $Y^m(n,n)$, we have $\frac{3}{4}n^2-\frac{3}{2}n+5$ faces including three openings (where the tubes meet to the amchair) each of chirality $(n, n)$, six heptagones, and $\frac{3}{4}n^2-\frac{3}{2}n-4$ hexagones. In addition, the tube $T_m(n,n)$ contains $2mn$ hexagonal faces.\\
 
 Let $n$, $m$, and $l$ be positive integers with $m\geq1$ and $n=2l$, for some $l\geq2$. Then $J=J^m(n,n)$ be the $Y$-junction graph of $Y^m(n,n)$. It has $9l^2-3l+2$ hexagonal rings along with six heptagons. The graph $J$ is of order $6l^2 + 18l + 6 + 24ml$ and size $9l^2 + 21l + 9 + 36ml$. It has $6l^2+12l+6+24ml$ vertices of degree three and $12l$ vertices of degree two. Note that graph $J$ is a $2$-conneced graph.\\
 
 Along with $2$-connected Y-junction graph $J$, the $1$-connected Y-junction graphs have also been taken into consideration. These graphs are obtained by adding pendants to the degree $2$ vertices of the $2$-connected graph $J$. Note that, each tube of $J$ has $2n$ vertices of degree $2$. Therefore, the graph $J$ has $6n$ vertices of degree $2$.\\
 
  The graph obtained by connecting $2n$ pendants to any one tube in $J$ is denoted by $J_1$, and we call it as second type Y-junction graph. The order and size of graph $J_1$ are $6l^2 +22l +6+24ml$ and $9l^2 +25l +9+36ml$, respectively. The graph $J_2$ represents a graph which is obtained by attaching $4n$ pendants to any two tubes of $J$ and we call it as third type Y-junction graph. In $J_2$, we have $6l^2+26l+6+24ml$ vertices and $9l^2+29l+9+36ml$ edges. The graph obtained by joining $6n$ pendants to all the three tubes of $J$ is denoted by $J_3$, and we called it as fourth type Y-junction graph. It has $6l^2 +30l +6+24ml$ vertices and $9l^2 +33l +9+36ml$ edges. The carbon nanotube Y-junction graphs $J$, $J_1$, $J_2$, and $J_3$ are shown in Figure $1$.\\
  
   The edge partition of Y-junction graphs $J$, $J_1$, $J_2$, and $J_3$ based on the neighborhood degree-sum of end vertices of an edge is given in Table $2$.
   \begin{figure}
   	\centering
   	\begin{subfigure}[b]{0.5\textwidth}
   		\centering
   		\includegraphics[width=\textwidth]{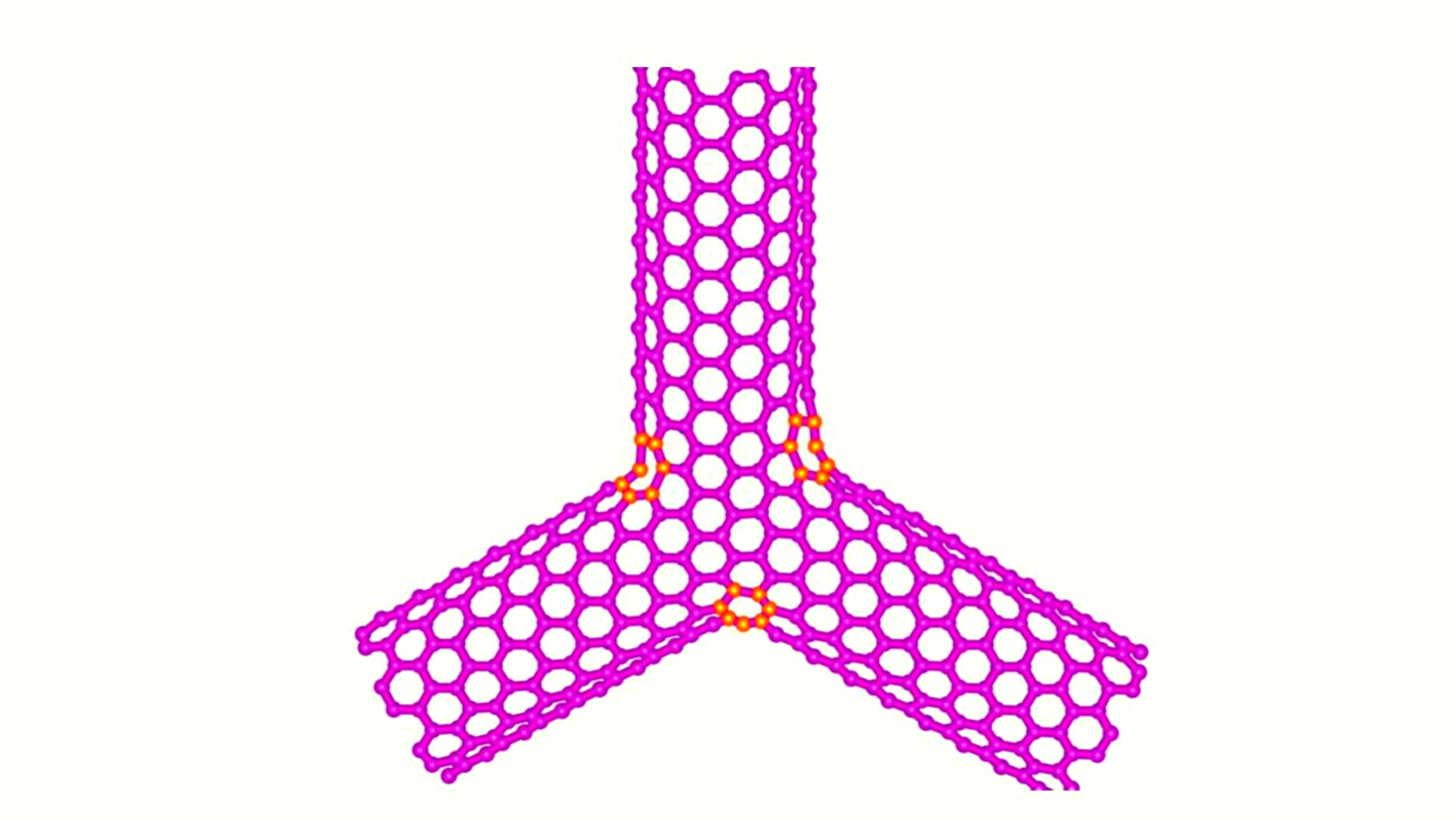}
   		\caption{Y-junction graph $J$ }
   	\end{subfigure}
   	\begin{subfigure}[b]{0.35\textwidth}
   		\centering
   		\includegraphics[width=\textwidth]{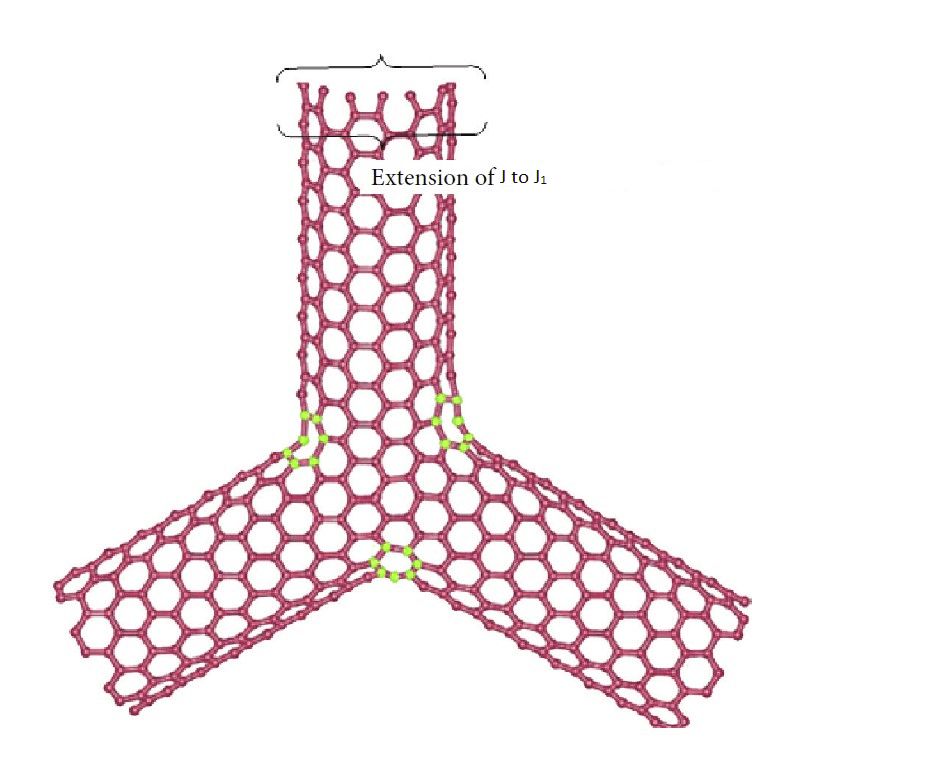}
   		\caption{Y-junction graph $J_1$ }
   	\end{subfigure}
   	\begin{subfigure}[b]{0.3\textwidth}
   		\centering
   		\includegraphics[width=\textwidth]{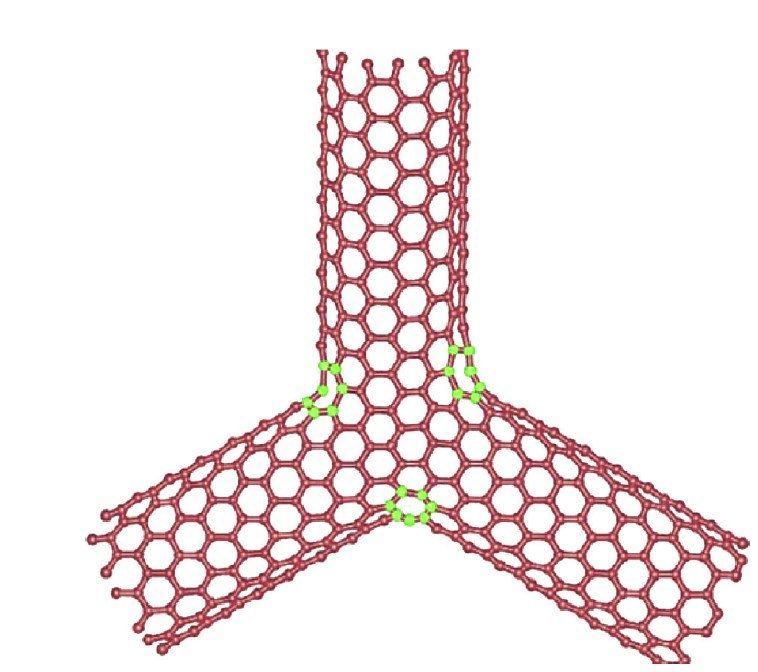}
   		\caption {Y-junction graph $J_2$ }
   	\end{subfigure}
   	\hspace{2cm}
   	\begin{subfigure}[b]{0.29\textwidth}
   		\centering
   		\includegraphics[width=\textwidth]{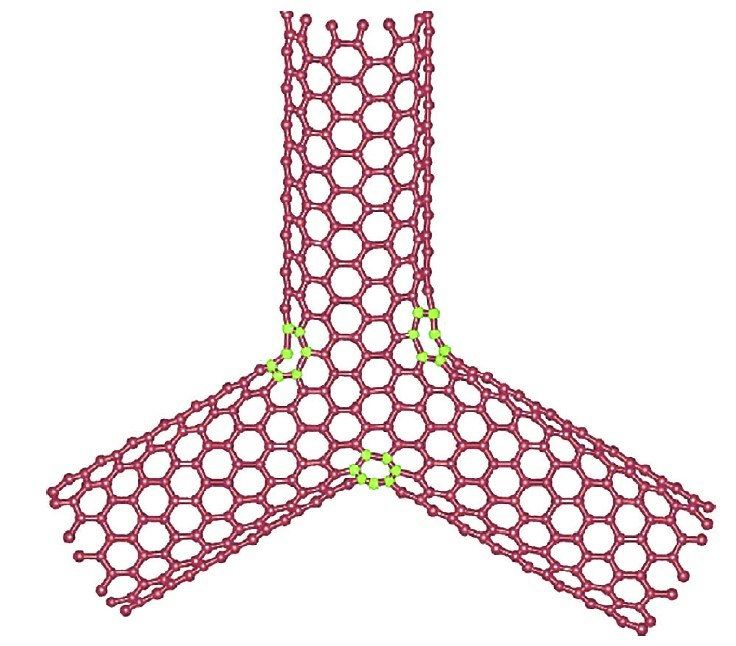}
   		\caption{Y-junction graph $J_3$ }
   	\end{subfigure}
   	\caption{A symmetrical uncapped
   		single-walled armchair carbon
   		nanotubes Y-junction graphs}
   	\label{fig2}
   \end{figure}
\begin{table}[H]
	\centering
	\caption{Edge partitions of $J$, $J_1$, $J_2$, and $J_3$}
	\setlength{\tabcolsep}{8pt}
	\renewcommand{\arraystretch}{2.0}
	\begin{tabular}{l cccc}
		\hline
		$d_n(u),d_n(v)$ & $J$-frequency & $J_1$-frequency&$J_2$-frequency&$J_3$-frequency\\
		\hline
		(3,7) & 0 & 4l &8l &12l\\
		(5,5) & 6l & 4l &2l &0\\
		(5,8) & 12l & 8l &4l &0\\
		(7,7) & 0 & 2l &4l &6l\\
		(7,9) & 0 & 4l &8l &12l\\
		(8,8) & 6l & 4l &2l &0\\
		(8,9) & 12l & 8l &4l &0\\
		(9,9) & $9l^2-15l+36ml+9$ & $9l^2-9l+36ml+9$ &$9l^2-3l+36ml+9$ &$9l^2+3l+36ml+9$\\
		\hline
		
	\end{tabular}
\end{table}
\section{NM-Polynomials and Topological Indices of Y-Junction Graphs}
In this section, we develop the general expression of NM-polynomials for the Y-junction graphs and then recover various neighborhood degree-sum based topological indices from these polynomials.  

\begin{thm}
	Let $J$ be the Y-junction graph of an uncapped symmetrical single-walled armchair carbon nanotube. Then\\
	$NM(J;x,y)=6lx^5y^5+12lx^5y^8+6lx^8y^8+12lx^8y^9+(9l^2-15l+9+36ml)x^9y^9$.
\end{thm}
\begin{proof}
	The Y-junction graph of an uncapped symmetrical single-walled armchair carbon nanotubes has $9l^2+21l+9+36ml$ number of edges. Let $E_{(i,j)}$ be the set of all edges with neighborhood degree sum of end vertices $i,j$, i.e.,
	$E_{(i,j)}=\{uv\in E(J):d_n(u)=i,\; d_n(v)=j\}$.\\
	By means of structural analysis of $J$, the edge set of $J$ can be partitioned into five sets on the basis of neighborhood degree sum of end vertices as follows:\\
	$E_{(5,5)}=\{uv\in E(J):d_n(u)=5,\; d_n(v)=5\}$, $E_{(5,8)}=\{uv\in E(J):d_n(u)=5,\; d_n(v)=8\}$, $E_{(8,8)}=\{uv\in E(J^m(n,n)):d_n(u)=8,\; d_n(v)=8\}$, $E_{(8,9)}=\{uv\in E(J):d_n(u)=8,\; d_n(v)=9\}$, $E_{(9,9)}=\{uv\in E(J):d_n(u)=9,\; d_n(v)=9\}$, and $\lvert E_{(5,5)} \rvert=6l$, $\lvert E_{(5,8)} \rvert=12l$, $\lvert E_{(8,8)} \rvert=6l$, $\lvert E_{(8,9)} \rvert=12l$, $\lvert E_{(9,9)} \rvert=9l^2-15l+9+36ml$.\\
	
	From Equation (1), the NM-polynomial of $J$ is obtained as follows:
	\begin{eqnarray*}
		NM(J;x,y)&=&\sum\limits_{i\leq j}\lvert E_{(i,j)}\rvert x^iy^j\\
		&=&\lvert E_{(5,5)}\rvert x^5y^5+\lvert E_{(5,8)}\rvert x^5y^8+\lvert E_{(8,8)}\rvert x^8y^8+\lvert E_{(8,9)}\rvert x^8y^9+\lvert E_{(9,9)}\rvert x^9y^9\\
		&=&6lx^5y^5+12lx^5y^8+6lx^8y^8+12lx^8y^9+(9l^2-15l+9+36ml)x^9y^9.
	\end{eqnarray*}
\end{proof}
\begin{thm}
	Let $J$ be the Y-junction graph of an uncapped symmetrical single-walled armchair carbon nanotube . Then\\\\(i) $NM_1(J)=162l^2+246l+648ml+162$\\\\
	(ii) $NM_2(J)=729l^2+663l+2916ml+729$\\\\
	(iii) $NF(J)=1458l^2+1446l+5832ml+1458$\\\\
	(iv) $^{nm}M_2(J)=0.11l^2+0.62l+0.44ml+0.11$\\\\
	(v) $NR_{\alpha}(J)=6l(25^{\alpha}+2(40)^{\alpha}+64^{\alpha}+2(72)^{\alpha})+81^{\alpha}(9l^2-15l+9+36ml)$\\\\
	(vi) $ND_3(J)=13122l^2+6702l+52488ml+13122$\\\\
	(vii) $ND_5(J)=18l^2+44.86l+72ml+18$\\\\
	(viii) $NH(J)=l^2+9.69l+4ml+1 $\\\\
	(ix) $NI(J)=40.5l^2+59.24l+162ml+40.5$\\\\
	(x) $S(J)=1167.7l^2+714.23l+4670.9ml+1167.7$.
\end{thm}

\begin{proof}
	Let $f(x,y)=NM(J;x,y)=6lx^5y^5+12lx^5y^8+6lx^8y^8+12lx^8y^9+(9l^2-15l+9+36ml)x^9y^9$. Then, we have\\
	\noindent$D_x(f(x,y))=30lx^5y^5+60lx^5y^8+48lx^8y^8+96lx^8y^9+9(9l^2-15l+9+36ml)x^9y^9$.\\
	
	\noindent$D_y(f(x,y))=30lx^5y^5+96lx^5y^8+48lx^8y^8+108lx^8y^9+9(9l^2-15l+9+36ml)x^9y^9$.\\
	
	\noindent$D^2_x(f(x,y))=150lx^5y^5+300lx^5y^8+384lx^8y^8+768lx^8y^9+81(9l^2-15l+9+36ml)x^9y^9$.\\

	\noindent$D^2_y(f(x,y))=150lx^5y^5+768lx^5y^8+384lx^8y^8+972lx^8y^9+81(9l^2-15l+9+36ml)x^9y^9$.\\

	\noindent$D_xD_y(f(x,y))=150lx^5y^5+480lx^5y^8+384lx^8y^8+864lx^8y^9+81(9l^2-15l+9+36ml)x^9y^9$.\\

	\noindent$(D_x+D_y)f(x,y)=60lx^5y^5+156lx^5y^8+96lx^8y^8+204lx^8y^9+18(9l^2-15l+9+36ml)x^9y^9$.
	\begin{eqnarray*}D_xD_y(D_x+D_y)f(x,y)&=&1500lx^5y^5+6240lx^5y^8+6144lx^8y^8+14688lx^8y^9+1458(9l^2-15l+\\&&9+36ml)x^9y^9.
		\end{eqnarray*}

	\noindent$(D^2_x+D^2_y)f(x,y)=300lx^5y^5+1068lx^5y^8+768lx^8y^8+1740lx^8y^9+162(9l^2-15l+9+36ml)x^9y^9$.
	\begin{eqnarray*}
	D_x^{\alpha}D_y^{\alpha}(f(x,y))&=&6l(25)^{\alpha
	}x^5y^5+12l(40)^{\alpha}x^5y^8+6l(64)^{\alpha}x^8y^8+12l(72)^{\alpha}x^8y^9+(81)^{\alpha}\\&&(9l^2-15l+9+36ml)x^9y^9.	\end{eqnarray*}
\noindent$S_xS_y(f(x,y))=\frac{6l}{25}x^5y^5+\frac{12l}{40}x^5y^8+\frac{6l}{64}x^8y^8+\frac{12l}{72}x^8y^9+\frac{(9l^2-15l+9+36ml)}{81}x^9y^9$.\\

\noindent$S_yD_x+S_xD_y(f(x,y))=12lx^5y^5+\frac{267l}{10}x^5y^8+12lx^8y^8+\frac{145l}{6}x^8y^9+2(9l^2-15l+9+36ml)x^9y^9$.\\

\noindent$2S_xT(f(x,y))=\frac{6l}{5}x^{10}+\frac{24l}{13}x^{13}+\frac{3l}{4}x^{16}+\frac{24l}{17}x^{17}+\frac{(9l^2-15l+9+36ml)}{9}x^{18}$.\\

\noindent$S_xTD_xD_y(f(x,y))=15lx^{10}+\frac{480l}{13}x^{13}+\frac{384l}{16}x^{16}+\frac{864l}{17}x^{17}+\frac{81(9l^2-15l+9+36ml)}{18}x^{18}$.\\

\noindent$S_x^3Q_{-2}TD_x^3D_y^3(f(x,y))=\frac{93750l}{512}x^8+\frac{768000l}{1331}x^{11}+\frac{1572864l}{2744}x^{14}+\frac{4478976l}{3375}x^{15}+\frac{531441(9l^2-15l+9+36ml)}{4096}x^{16}$.\\

\noindent Now, using Table $1$ we have\\

(i) $NM_1(J)=(D_x+D_y)f(x,y)\rvert_{x=y=1}=162l^2+246l+648ml+162$.\\

(ii) $NM_2(J)=(D_xD_y)f(x,y)\rvert_{x=y=1}=729l^2+663l+2916ml+729$.\\

(iii) $NF(J)=(D^2_x+D^2_y)f(x,y)\rvert_{x=y=1}=1458l^2+1446l+5832ml+1458$.\\

(iv) $^{nm}M_2(J)=(S_xS_y)f(x,y)\rvert_{x=y=1}=0.11l^2+0.62l+0.44ml+0.11$.\\

(v) $NR_{\alpha}(J)=(D^{\alpha}_xD^{\alpha}_y)f(x,y)\rvert_{x=y=1}=6l(25^{\alpha}+2(40)^{\alpha}+64^{\alpha}+2(72)^{\alpha})+81^{\alpha}(9l^2-15l+9+36ml)$.\\

(vi) $ND_3(J)=D_xD_y(D_x+D_y)f(x,y)\rvert_{x=y=1}=13122l^2+6702l+52488ml+13122$.\\

(vii) $ND_5(J)=S_yD_x+S_xD_y(f(x,y))\rvert_{x=y=1}=18l^2+44.86l+72ml+18$.\\

(viii) $NH(J)=2S_xT(f(x,y))\rvert_{x=y=1}=l^2+9.69l+4ml+1 $.\\

(ix) $NI(J)=S_xTD_xD_y(f(x,y))\rvert_{x=y=1}=40.5l^2+59.24l+162ml+40.5$.\\

(x) $S(J)=S_x^3Q_{-2}TD_x^3D_y^3(f(x,y))\rvert_{x=y=1}=1167.7l^2+714.23l+4670.9ml+1167.7$.
\end{proof}
\begin{thm}
	Let $J_1$ be the second type Y-junction graph of an uncapped symmetrical single-walled armchair carbon nanotube. Then\\
	$NM(J_1;x,y)=4lx^3y^7+4lx^5y^5+8lx^5y^8+2lx^7y^7+4lx^7y^9+4lx^8y^8+8lx^8y^9+(9l^2-9l+9+36ml)x^9y^9$.
\end{thm}
\begin{proof}
	The second type Y-junction graph of an uncapped symmetrical single-walled armchair carbon nanotubes has $9l^2 +25l +9+36ml$ edges. Let $E_{(i,j)}$ be the set of all edges with neighborhood degree sum of end vertices $i,j$, i.e.,
	$E_{(i,j)}=\{uv\in E(J_1):d_n(u)=i,\; d_n(v)=j\}$.\\
	By means of structure analysis of $J_1$, the edge set of $J_1$ can be partitioned into eight sets on the basis of neighborhood degree sum of end vertices as follows:\\
	$E_{(3,7)}=\{uv\in E(J_1):d_n(u)=3,\; d_n(v)=7\}$, $E_{(5,5)}=\{uv\in E(J_1):d_n(u)=5,\; d_n(v)=5\}$, $E_{(5,8)}=\{uv\in E(J_1):d_n(u)=5,\; d_n(v)=8\}$, $E_{(7,7)}=\{uv\in E(J_1):d_n(u)=7,\; d_n(v)=7\}$, $E_{(7,9)}=\{uv\in E(J_1):d_n(u)=7,\; d_n(v)=9\}$, $E_{(8,8)}=\{uv\in E(J_1):d_n(u)=8,\; d_n(v)=8\}$, $E_{(8,9)}=\{uv\in E(J_1):d_n(u)=8,\; d_n(v)=9\}$, $E_{(9,9)}=\{uv\in E(J_1):d_n(u)=9,\; d_n(v)=9\}$, and $\lvert E_{(3,7)}\rvert=4l$, $\lvert E_{(5,5)}\rvert=4l$, $\lvert E_{(5,8)}\rvert=8l$, $\lvert E_{(7,7)}\rvert=2l$, $\lvert E_{(7,9)}\rvert=4l$, $\lvert E_{(8,8)}\rvert=4l$, $\lvert E_{(8,9)}\rvert=8l$, $\lvert E_{(9,9)}\rvert=9l^2-9l+9+36ml$. \\
	From Equation (1), the NM-polynomial of $J_1$ is obtained as follows:
	\begin{eqnarray*}
		NM(J_1;x,y)&=&\sum\limits_{i\leq j}\lvert E_{(i,j)}\rvert x^iy^j\\
		&=&\lvert E_{(3,7)}\rvert x^3y^7+\lvert E_{(5,5)}\rvert x^5y^5+\lvert E_{(5,8)}\rvert x^5y^8+\lvert E_{(7,7)}\rvert x^7y^7+\lvert E_{(7,9)}\rvert x^7y^9+\\&&\lvert E_{(8,8)}\rvert x^8y^8+\lvert E_{(8,9)}\rvert x^8y^9+\lvert E_{(9,9)}\rvert x^9y^9\\
		&=&4lx^3y^7+4lx^5y^5+8lx^5y^8+2lx^7y^7+4lx^7y^9+4lx^8y^8+8lx^8y^9+\\&&(9l^2-9l+9+36ml)x^9y^9.
	\end{eqnarray*}
\end{proof}
\begin{thm}
	Let $J_1$ be the second type Y-junction graph of an uncapped symmetrical single-walled armchair carbon nanotube. Then\\(i) $NM_1(J_1)=162l^2+314l+648ml+162$\\\\
	(ii) $NM_2(J_1)=729l^2+957l+2916ml+729$\\\\
	(iii) $NF(J_1)=1458l^2+2074l+5832ml+1458$\\\\
	(iv) $^{nm}M_2(J_1)=0.11l^2+0.72l+0.44ml+0.11$
		 \begin{eqnarray*}
	 (v) \; NR_{\alpha}(J_1)&=&2l(2(21)^\alpha+2(25)^\alpha+4(40)^\alpha+(49)^\alpha+2(63)^\alpha+2(64)^\alpha+4(72)^\alpha)+(81)^\alpha(9l^2-9l+9\\&&+36ml)\end{eqnarray*}
	(vi) $ND_3(J_1)=13122l^2+12170l+52488ml+13122$\\\\
	(vii) $ND_5(J_1)=18l^2+56.328l+72ml+18$\\\\
	(viii) $NH(J_1)=l^2+3.98l+4ml+1$\\\\
	(ix) $NI(J_1)=40.5l^2+75.15l+162ml+40.5$\\\\
	(x) $S(J_1)=1167.7l^2+1178.92l+4670.9ml+1167.7$.
\end{thm}
\begin{proof}
	Refer to Theorem $2$ for proof.
\end{proof}
\begin{thm}
	Let $J_2$ be the third type Y-junction graph of an uncapped symmetrical single-walled armchair carbon nanotube. Then\\
	$NM(J_2;x,y)=8lx^3y^7+2lx^5y^5+4lx^5y^8+4lx^7y^7+8lx^7y^9+2lx^8y^8+4lx^8y^9+(9l^2-3l+9+36ml)x^9y^9$.
\end{thm}
\begin{proof}
The third type Y-junction graph of an uncapped symmetrical single-walled armchair carbon nanotubes has $9l^2+29l+9+36ml$ number of edges. Let $E_{(i,j)}$ be the set of all edges with neighborhood degree sum of end vertices $i,j$, i.e.,
	$E_{(i,j)}=\{uv\in E(J_2):d_n(u)=i,\; d_n(v)=j\}$.\\
	By means of structure analysis of $J_2$, the edge set of $J_2$ can be partitioned into eight sets on the basis of neighborhood degree sum of end vertices as follows:\\
	$E_{(3,7)}=\{uv\in E(J_2):d_n(u)=3,\; d_n(v)=7\}$, $E_{(5,5)}=\{uv\in E(J_2):d_n(u)=5,\; d_n(v)=5\}$, $E_{(5,8)}=\{uv\in E(J_2^):d_n(u)=5,\; d_n(v)=8\}$, $E_{(7,7)}=\{uv\in E(J_2):d_n(u)=7,\; d_n(v)=7\}$, $E_{(7,9)}=\{uv\in E(J_2):d_n(u)=7,\; d_n(v)=9\}$, $E_{(8,8)}=\{uv\in E(J_2):d_n(u)=8,\; d_n(v)=8\}$, $E_{(8,9)}=\{uv\in E(J_2):d_n(u)=8,\; d_n(v)=9\}$, $E_{(9,9)}=\{uv\in E(J_2):d_n(u)=9,\; d_n(v)=9\}$, and $\lvert E_{(3,7)}\rvert=8l$, $\lvert E_{(5,5)}\rvert=2l$, $\lvert E_{(5,8)}\rvert=4l$, $\lvert E_{(7,7)}\rvert=4l$, $\lvert E_{(7,9)}\rvert=8l$, $\lvert E_{(8,8)}\rvert=2l$, $\lvert E_{(8,9)}\rvert=4l$, $\lvert E_{(9,9)}\rvert=9l^2-3l+9+36ml$. \\
	
	From Equation (1), the NM-polynomial of $J_2$ is obtained as follows:
	\begin{eqnarray*}
		NM(J_2;x,y)&=&\sum\limits_{i\leq j}\lvert E_{(i,j)}\rvert x^iy^j\\
		&=&\lvert E_{(3,7)}\rvert x^3y^7+\lvert E_{(5,5)}\rvert x^5y^5+\lvert E_{(5,8)}\rvert x^5y^8+\lvert E_{(7,7)}\rvert x^7y^7+\lvert E_{(7,9)}\rvert x^7y^9+\\&&\lvert E_{(8,8)}\rvert x^8y^8+\lvert E_{(8,9)}\rvert x^8y^9+\lvert E_{(9,9)}\rvert x^9y^9\\
		&=&8lx^3y^7+2lx^5y^5+4lx^5y^8+4lx^7y^7+8lx^7y^9+2lx^8y^8+4lx^8y^9+\\&&(9l^2-3l+9+36ml)x^9y^9.
	\end{eqnarray*}
\end{proof}
\begin{thm}
	Let $J_2$ be the third type Y-junction graph of an uncapped symmetrical single-walled armchair carbon nanotube. Then\\\\(i) $NM_1(J_2)=162l^2+382l+648ml+162$\\\\
	(ii) $NM_2(J_2)=729l^2+1251l+2916ml+729$\\\\
	(iii) $NF(J_2)=1458l^2+2478l+5832ml+1458$\\\\
	(iv) $^{nm}M_2(J_2)=0.11l^2+0.819l+0.44ml+0.11$\\\\
	(v) $NR_{\alpha}(J_2)=2l(4(21)^\alpha+(25)^\alpha+2(40)^\alpha+2(49)^\alpha+4(63)^\alpha+(64)^\alpha+2(72)^\alpha)+(81)^\alpha(9l^2-3l+9+36ml)$\\\\
	(vi) $ND_3(J_2)=13122l^2+17638l+52488ml+13122$\\\\
	(vii) $ND_5(J_2)=18l^2+65.56l+72ml+18$\\\\
	(viii) $NH(J_2)=l^2+4.57l+4ml+1$\\\\
	(ix) $NI(J_2)=40.5l^2+91.048l+162ml+40.5$\\\\
	(x) $S(J_2)=1167.7l^2+1643.61l+4670.9ml+1167.7$.
\end{thm}
\begin{proof}
	Refer to Theorem $2$ for proof.
\end{proof}
\begin{thm}
	Let $J_3$ be the fourth type Y-junction graph of an uncapped symmetrical single-walled armchair carbon nanotube. Then\\
	$NM(J_3;x,y)=12lx^3y^7+6lx^7y^7+12lx^7y^9+(9l^2+3l+9+36ml)x^9y^9$.
\end{thm}
\begin{proof}
	The fourth type Y-junction graph of an uncapped symmetrical single-walled armchair carbon nanotube has $9l^2+33l+9+36ml$ number of edges. Let $E_{(i,j)}$ be the set of all edges with neighborhood degree sum of end vertices $i,j$, i.e.,
	$E_{(i,j)}=\{uv\in E(J_3):d_n(u)=i,\; d_n(v)=j\}$.\\
	By means of structure analysis of $J_3$, the edge set of $J_3$ can be partitioned into four sets on the basis of neighborhood degree sum of end vertices as follows:\\
	$E_{(3,7)}=\{uv\in E(J_3):d_n(u)=3,\;d_n(v)=7\}$, $E_{(7,7)}=\{uv\in E(J_3):d_n(u)=7,\;d_n(v)=7\}$, $E_{(7,9)}=\{uv\in E(J_3):d_n(u)=7,\;d_n(v)=9\}$, $E_{(9,9)}=\{uv\in E(J_3):d_n(u)=9,\;d_n(v)=9\}$, and $\lvert E_{(3,7)}\rvert=12l$, $\lvert E_{(7,7)}\rvert=6l$, $\lvert E_{(7,9)}\rvert=12l$, $\lvert E_{(9,9)}\rvert=9l^2+3l+9+36ml$.\\
	
	From Equation (1), the NM-polynomial of $J_3$ is obtained as follows:
		\begin{eqnarray*}
		NM(J_3;x,y)&=&\sum\limits_{i\leq j}\lvert E_{(i,j)}\rvert x^iy^j\\
		&=&\lvert E_{(3,7)}\rvert x^3y^7+\lvert E_{(7,7)}\rvert x^7y^7+\lvert E_{(7,9)}\rvert x^7y^9+\lvert E_{(9,9)}\rvert x^9y^9\\
		&=&12lx^3y^7+6lx^7y^7+12lx^7y^9+(9l^2+3l+9+36ml)x^9y^9.
	\end{eqnarray*}
\end{proof}
\begin{thm}
	Let $J_3$ be the fourth type Y-junction graph of an uncapped symmetrical single-walled armchair carbon nanotube. Then\\\\(i) $NM_1(J_3)=162l^2+450l+648ml+162$\\\\
	(ii) $NM_2(J_3)=729l^2+1545l+2916ml+729$\\\\
	(iii) $NF(J_3)=1458l^2+3330l+5832ml+1458$\\\\
	(iv) $^{nm}M_2(J_3)=0.11l^2+0.92l+0.44ml+0.11$\\\\
	(v) $NR_{\alpha}(J_3)=6l(2(21)^\alpha+(49)^\alpha+2(63)^\alpha)+(81)^\alpha(9l^2+3l+9+36ml)$\\\\
	(vi) $ND_3(J_3)=13122l^2+23106l+52488ml+13122$\\\\
	(vii) $ND_5(J_3)=18l^2+75.90l+72ml+18$\\\\
	(viii) $NH(J_3)=l^2+5.090l+4ml+1$\\\\
	(ix) $NI(J_3)=40.5l^2+106.95l+162ml+40.5$\\\\
	(x) $S(J_3)=1167.7l^2+2085.95l+4670.9ml+1167.7$.
\end{thm}
\begin{proof}
	Refer to Theorem $2$ for proof. 
\end{proof}
\section{Graph Index-Entropies of Y-Junction Graphs} 
In this section, we compute the index-entropy of carbon nanotube Y-junctions in terms of neighborhood degree sum-based topological indices. We first compute index-entropies of the Y-junction graph $J$ whose edge partition is given in Table $2$.
\begin{itemize}
\item \textbf{Third-version of Zagreb index-entropy of $J$}
\end{itemize}
From part (i) of Theorem $2$, we have
\begin{equation}
	 NM_1(J)=162l^2+246l+648ml+162.
 \end{equation}
Now, from Equation ($4$), the third-version of Zagreb index-entropy of $J$ is  \begin{equation}
	H_{\beta_1}(J)=log(NM_1(J))-\frac{1}{NM_1(J)}\sum\limits_{e\in E(J)}{\beta}_1(e)log\beta_1(e).
\end{equation}
Using Table $2$ and Equation ($19$) in Equation ($20$), we get the required third-version of Zagreb index-entropy of $J$ as follows: 
\begin{eqnarray*}
	H_{\beta_1}(J)&=&log(NM_1(J))-\frac{1}{NM_1(J)}\sum\limits_{e\in E(J)}{\beta}_1(e)log\beta_1(e)\\
	&=& log(162l^2+246l+648ml+162)-\frac{1}{162l^2+246l+648ml+162}\bigg(6l(10)(log10)+\\&&12l(13)(log13)+6l(16)(log16)+12l(17)(log17)+(9l^2-15l+36ml+9)(18)(log18)\bigg)\\
	&=&log(162l^2+246l+648ml+162)-\frac{1}{162l^2+246l+648ml+162}\bigg(60l(log10)+\\&&156l(log13)+96l(log16)+204l(log17)+(162l^2-270l+648ml+162)(log18)\bigg)\\
	&=&log(162l^2+246l+648ml+162)-\frac{1}{162l^2+246l+648ml+162}\bigg(60l(1)+156l(1.1139433523)\\&&+96l(1.2041199827)+204l(1.2304489214)+(162l^2-270l+648ml+162)(1.2552725051)\bigg)\\&&
	\approx log(162l^2+246l+648ml+162)-\frac{202.5l^2+261.78l+810ml+202.5}{162l^2+246l+648ml+162}.
\end{eqnarray*}
\begin{itemize}
	\item \textbf{Neighborhood second Zagreb index-entropy of $J$}\end{itemize}
From part (ii) of Theorem $2$, we have 
\begin{equation}
	NM_2(J)=729l^2+663l+2916ml+729.
\end{equation}
By using the values given in Table $2$ and Equation ($21$) in Equation ($6$), we get the required neighborhood second Zagreb index-entropy of $J$ as follows:
\begin{eqnarray*}
	H_{\beta_2}(J)&=&log(NM_2(J))-\frac{1}{NM_2(J)}\sum\limits_{e\in E(J)}{\beta}_2(e)log\beta_2(e)\\
	&=&log(729l^2+663l+2916ml+729)-\frac{1}{729l^2+663l+2916ml+729}\bigg(6l(25)(log25)+\\&&12l(40)(log40)+6l(64)(log64)+12l(72)(log72)+(9l^2-15l+36ml+9)(81)(log81)\bigg)\\&&
	\approx log(729l^2+663l+2916ml+729)-\frac{1391.22l^2+958.27l+5564.88ml+1391.22}{729l^2+663l+2916ml+729}.
\end{eqnarray*}
Similarly, we compute the remaning  index-entropies of $J$. Table $3$ shows some calculated graph index-entropies of $J$. \\

In this way, the topological index-based entropies for Y-junction graphs $J_1$, $J_2$, and $J_3$ are calculated.\\
 The index-based entropies of $J_1$  $J_2$, and $J_3$ are given in Tables $4$, $5$, and $6$.
\begin{table}[H]
	\centering
	\caption{Index-entropies of $J$}
	\setlength{\tabcolsep}{10pt}
	\renewcommand{\arraystretch}{1.9}
	\begin{tabular}{l cc}
		\hline
		Entropy & Values of entropies\\
		\hline
	$H_{\beta_3}(J)$	& log$(1458l^2+1446l+5832ml+1458)-\frac{3221.62l^2+2601.62l+12885.84ml+3221.46}{1458l^2+1446l+5832ml+1458}$\\
	$H_{\beta_4}(J)$	& log$(0.11l^2+0.62l+0.44ml+0.11)+\frac{0.207l^2+0.92l+0.828ml+0.207}{0.11l^2+0.62l+0.44ml+0.11}$\\
	$H_{\beta_5}(J)$	& log$(13122l^2+12170l+52488ml+13122)-\frac{41514.75l^2+15202.13l+166059.31ml+41514.75}{13122l^2+12170l+52488ml+13122}$\\
	$H_{\beta_6}(J)$	& log$(18l^2+44.86l+72ml+18)-\frac{5.41l^2+14.81l+21.67ml+5.41}{18l^2+44.86l+72ml+18}$\\
	$H_{\beta_7}(J)$	& log$(l^2+9.69l+4ml+1)+\frac{0.95l^2+2.72l+3.81ml+0.95}{l^2+9.69l+4ml+1}$\\
	$H_{\beta_8}(J)$	& log$(40.5l^2+59.24l+162ml+40.5)-\frac{26.45l^2+26.18l+105.82ml+26.45}{40.5l^2+59.24l+162ml+40.5}$\\
	\hline
	\end{tabular}
\end{table} 

\begin{table}[H]
	\centering
	\caption{Index-entropies of $J_1$}
	\setlength{\tabcolsep}{10pt}
	\renewcommand{\arraystretch}{1.9}
	\begin{tabular}{l cc}
		\hline
		Entropy & Values of entropies\\
		\hline
		$H_{\beta_1}(J_1)$	& log$(162l^2+314l+648ml+162)-\frac{203.31l^2+346.09l+813.24ml+203.31}{162l^2+314l+648ml+162}$\\
		$H_{\beta_2}(J_1)$	& log$(729l^2+957l+2916ml+729)-\frac{1391.22l^2+1523.54l+5564.88ml+1391.22}{729l^2+957l+2916ml+729}$\\
		$H_{\beta_3}(J_1)$	& log$(1458l^2+2074l+5832ml+1458)-\frac{3221.46l^2+3991l+12885.84ml+3221.46}{1458l^2+2074l+5832ml+1458}$\\
		$H_{\beta_4}(J_1)$	& log$(0.11l^2+0.72l+0.44ml+0.11)+\frac{0.207l^2+1.065l+0.897ml+0.207}{0.11l^2+0.72l+0.44ml+0.11}$\\
		$H_{\beta_5}(J_1)$	& log$(13122l^2+12170l+52488ml+13122)-\frac{41514.75l^2+32699.53l+166059ml+41514.75}{13122l^2+12170l+52488ml+13122}$\\
		$H_{\beta_6}(J_1)$	& log$(18l^2+56.328l+72ml+18)-\frac{5.41l^2+19.09l+21.67ml+5.41}{18l^2+56.328l+72ml+18}$\\
		$H_{\beta_7}(J_1)$	& log$(l^2+3.98l+4ml+1)+\frac{0.9l^2+3.21l+3.6ml+0.9}{l^2+3.98l+4ml+1}$\\
		$H_{\beta_8}(J_1)$	& log$(40.5l^2+75.15l+162ml+40.5)-\frac{26.37l^2+36.84l+105.48ml+26.37}{40.5l^2+75.15l+162ml+40.5}$\\
		\hline
	\end{tabular}
\end{table}
\begin{table}[h]
	\centering
	\caption{Index-entropies of $J_2$}
	\setlength{\tabcolsep}{10pt}
	\renewcommand{\arraystretch}{1.9}
	\begin{tabular}{l cc}
		\hline
		Entropy & Values of entropies\\
		\hline
		$H_{\beta_1}(J_2)$	& log$(162l^2+382l+648ml+162)-\frac{203.31l^2+430.65l+813.24ml+203.31}{162l^2+382l+648ml+162}$\\
		$H_{\beta_2}(J_2)$	& log$(729l^2+1251l+2916ml+729)-\frac{1391.22l^2+2088.77l+5564.88ml+1391.22}{729l^2+1251l+2916ml+729}$\\
		$H_{\beta_3}(J_2)$	& log$(1458l^2+2478l+5832ml+1458)-\frac{3221.46l^2+5380.37l+12885.84ml+3221.46}{1458l^2+2478l+5832ml+1458}$\\
		$H_{\beta_4}(J_2)$	& log$(0.11l^2+0.819l+0.44ml+0.11)+\frac{0.099l^2+1.007l+0.396ml+0.099}{0.11l^2+0.819l+0.44ml+0.11}$\\
		$H_{\beta_5}(J_2)$	& log$(13122l^2+17638l+52488ml+13122)-\frac{41514.75l^2+50196.95l+166059ml+50196.95}{13122l^2+17638l+52488ml+13122}$\\
		$H_{\beta_6}(J_2)$	& log$(18l^2+65.56l+72ml+18)-\frac{5.41l^2+23.44l+21.67ml+5.41}{18l^2+65.56l+72ml+18}$\\
		$H_{\beta_7}(J_2)$	& log$(l^2+4.57l+4ml+1)+\frac{0.9l^2+3.61l+3.6ml+0.9}{l^2+4.57l+4ml+1}$\\
		$H_{\beta_8}(J_2)$	& log$(40.5l^2+91.048l+162ml+40.5)-\frac{26.37l^2+46.387l+105.48ml+26.37}{40.5l^2+91.048l+162ml+40.5}$\\
		\hline
	\end{tabular}
\end{table}
\begin{table}[H]
	\centering
	\caption{Index-entropies of $J_3$}
	\setlength{\tabcolsep}{10pt}
	\renewcommand{\arraystretch}{1.7}
	\begin{tabular}{l cc}
		\hline 
		Entropy & Values of entropies\\
		\hline
		$H_{\beta_1}(J_3)$	& log$(162l^2+450l+648ml+162)-\frac{203.31l^2+515.23l+813.24ml+203.31}{162l^2+450l+648ml+162}$\\
		$H_{\beta_2}(J_3)$	& log$(729l^2+1545l+2916ml+729)-\frac{1391.22l^2+2654.14l+5564.88ml+1391.22}{729l^2+1545l+2916ml+729}$\\
		$H_{\beta_3}(J_3)$	& log$(1458l^2+3330l+5832ml+1458)-\frac{3221.46l^2+6769.75l+12885.84ml+3221.46}{1458l^2+3330l+5832ml+1458}$\\
		$H_{\beta_4}(J_3)$	& log$(0.11l^2+0.92l+0.44ml+0.11)+\frac{0.18l^2+1.35l+0.72ml+0.18}{0.11l^2+0.92l+0.44ml+0.11}$\\
		$H_{\beta_5}(J_3)$	& log$(13122l^2+23106l+52488ml+13122)-\frac{41514.75l^2+67694.4l+166059ml+41514}{13122l^2+23106l+52488ml+13122}$\\
		$H_{\beta_6}(J_3)$	& log$(18l^2+75.90l+72ml+18)-\frac{5.41l^2+27.79l+21.67ml+5.41}{18l^2+75.90l+72ml+18}$\\
		$H_{\beta_7}(J_3)$	& log$(l^2+5.090l+4ml+1)+\frac{0.9l^2+4.04l+3.6ml+0.9}{l^2+5.090l+4ml+1}$\\
		$H_{\beta_8}(J_3)$	& log$(40.5l^2+106.95l+162ml+40.5)-\frac{26.37l^2+56.44l+105.48ml+26.37}{40.5l^2+106.95l+162ml+40.5}$\\
		\hline
	\end{tabular}
\end{table}
\section{Numerical Results and Discussions}
The numerical values of topological indices and graph index-entropies of Y-junction graphs are computed in this section for some values of $l$ and $m$. In addition, we plot line and bar graphs for comparison of the obtained results. Here, we use the logarithm of the base $10$ for calculations.\\

The numerical values of topological indices for Y-junction graph $J$ are given in Table $7$. The logarithmic values of Table $7$ are plotted in Figure $2$. From the vertical axis of Figure $2$, we can conclude that for Y-junction graph $J$, the topological indices have the following order: $^{nm}M_2\leq NR_{-1/2} \leq NH\leq ND_5\leq NI \leq NM_1\leq NM_2\leq S \leq NF \leq ND_3$. The third NDe index has the most dominating nature compared to other topological indices, whereas neighborhood second modified Zagreb index grew slowly.
\begin{table}[H]
	\centering
	\caption{Numerical values of topological indices for Y-junction graph $J$}
	\setlength{\tabcolsep}{6pt}
	\renewcommand{\arraystretch}{1.7}
	\tiny\begin{tabular}{lccccccccccc}
		\hline 
		[$l,m$]& $NM_1(J)$&$NM_2(J)$&$NF(J)$&$N^mM_2(J)$&$NR_{-\frac{1}{2}}(J)$&$ND_3(J)$&$ND_5(J)$&$NH(J)$&$NI(J)$&$S(J)$\\
		\hline
		[2,2]&3894&16635&33510&3.55&20.7436&288966&467.72&40.38&968.92&25950.56\\
		
		[3,3]&8190&35523&71406&6.92&45.27693&623718&962.58&75.07&2040.63&55857.79\\
		
		[4,4]&14106&61701&123882&11.39&79.81026&1089690&1637.44&119.76&3517.34&97442.22\\
		
		[5,5]&21642&95169&190938&16.96&124.3436&1686882& 2492.3&174.45&5399.05&150703.9\\
		
		[6,6]&30798&135927&272574&23.63&178.8769&2415294&3527.16&239.14&7685.76&215642.7\\
		
		[7,7] &41574&183975&368790&31.4&243.4103&3274926&4742.02&313.83&10377.47&292258.7\\
		
		[8,8]&53970&239313&479586&40.27&317.9436&4265778&6136.88&398.52&13474.18&380551.9\\
		
		[9.9]&67986&301941&604962&50.24&440.4769&5387850&7711.74&493.21&16975.89&480522.4\\
		
		[10,10]&83622&371859&744918&61.31&497.0103&6641142&9466.6&597.9&20882.6&592170\\
		\hline
	\end{tabular}
\end{table}
\begin{figure}[H]
	\centering
	\includegraphics[width=3.2in]{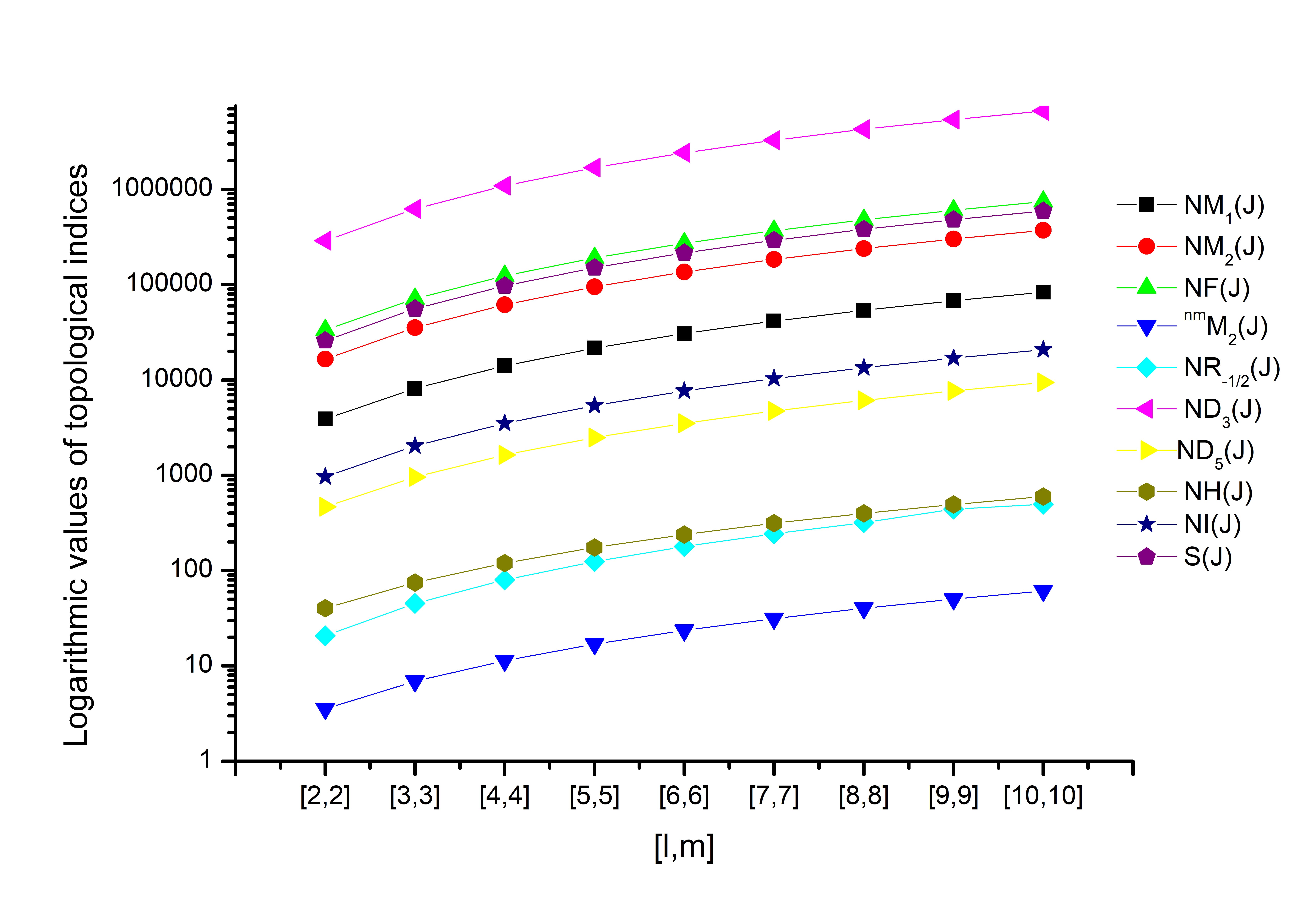}
	\caption{Graphical comparison among topological indices of Y-junction graph $J$ }
\end{figure}
Table $8$ shows some numerical values of topological indices for Y-junction graph $J_1$. The logarithmic values of these topological indices are plotted in Figure $3$. From Figure $3$, we can conclude that the topological indices for Y-junction graph $J_1$ have the following order: $^{nm} M_2\leq NH\leq NR_{-1/2}\leq ND_5\leq NI\leq NM_1\leq NM_2 \leq S \leq NF\leq  ND_3$. Also, we see that the logarithemic values of $NR_{-1/2}$ and $NH$ for $J_1$ are almost same.
\begin{table}[H]
	\centering
	\caption{Numerical values of topological indices for Y-junction graph $J_1$}
	\setlength{\tabcolsep}{6pt}
	\renewcommand{\arraystretch}{2}
	\tiny\begin{tabular}{lccccccccccc}
		\hline 
		[$l,m$]& $NM_1(J_1)$&$NM_2(J_1)$&$NF(J_1)$&$^		{nm}M_2(J_1)$&$NR_{-\frac{1}{2}}(J_1)$&$ND_3(J_1)$&$ND_5(J_1)$&$NH(J_1)$&$NI(J_1)$&$S(J_1)$\\
		\hline
		[2,2]&4030&17223&35166&3.75&29.34052&299902&490.656&28.96&1000.8&26879.94\\
		
		[3,3]&8394&36405&74190&7.22&58.51078&640122&996.984&57.94&2088.45&57251.86\\
		
		[4,4]&14378&62877&127994&11.79&97.68103&1111562&1683.312&96.92&3581.1&99300.98\\
		
		[5,5]&21982&96639&196578&17.46&146.8513&1714222& 2549.64&145.9&5478.75&153027.3\\
		
		[6,6]&31206&137691&279942&24.23&206.0216&2448102&3595.968&204.88&7781.4&218430.8\\
		
		[7,7] &42050&186033&378086&32.1&275.1918&3313202&4822.296&273.86&10489.05&295511.5\\
		
		[8,8]&54514&241665&491010&41.07&354.3621&4309522&6228.624&352.84&13601.7&384269.5\\
		
		[9.9]&68598&304587&618714&51.14&443.5323&5437062&7814.952&441.82&17119.35&484704.6\\
		
		[10,10]&84302&374799&761198&62.31&542.7026&6695822&9581.28&540.8&21042&596816.9\\
		\hline
	\end{tabular}
\end{table}
\begin{figure}[H]
	\centering
	\includegraphics[width=4in]{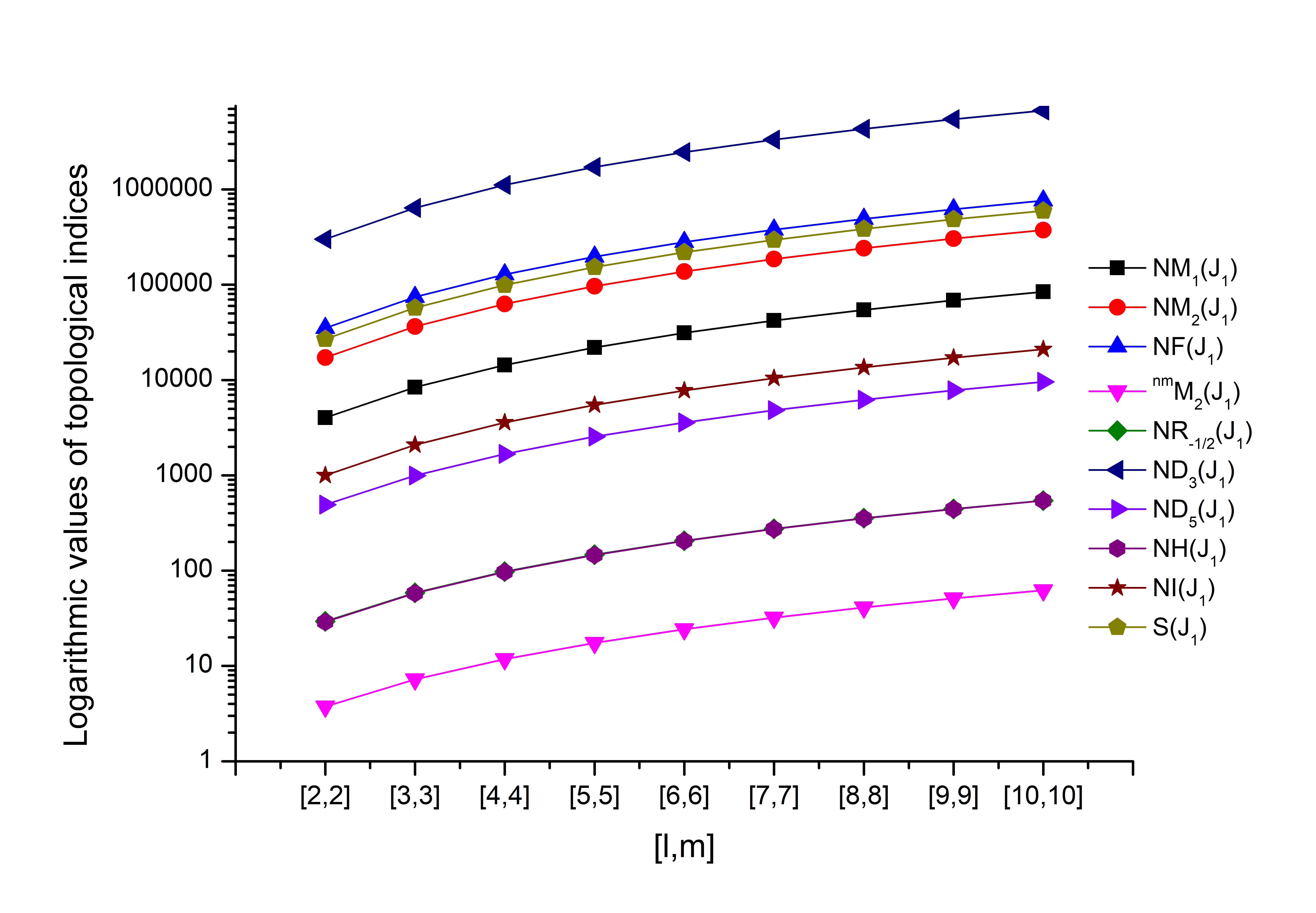}
	\caption{Graphical comparison among topological indices of Y-junction graph $J_1$ }
\end{figure}
Table $9$ shows some calculated values of topological indices for Y-junction graph $J_2$. The logarithmic values of these indices are plotted in Figure $4$. The vertical axis of Figure $4$ shows the comparison clearly. Figure $4$ shows that the logarithmic values of $ND_3$ are extremely high when compared to other topological indices of $J_2$. From Figure $4$, we see that the graph of $NR_{-1/2}$ and $NH$ are almost coincide.
\begin{table}[H]
	\centering
	\caption{Numerical values of topological indices for Y-junction graph $J_2$}
	\setlength{\tabcolsep}{6pt}
	\renewcommand{\arraystretch}{2}
	\tiny\begin{tabular}{lccccccccccc}
		\hline 
		[$l,m$]& $NM_1(J_2)$&$NM_2(J_2)$&$NF(J_2)$&$^{nm}M_2(J_2)$&$NR_{-\frac{1}{2}}(J_2)$&$ND_3(J_2)$&$ND_5(J_2)$&$NH(J_2)$&$NI(J_2)$&$S(J_2)$\\
		\hline
		[2,2]&4166&17811&35574&3.948&30.49121&310838&509.12&30.14&1032.596&27809.32\\
		
		[3,3]&8598&37287&74502&7.517&60.23681&656526&1024.68&59.71&2136.144&58645.93\\
		
		[4,4]&14650&64053&128010&12.186&99.98241&1133434&1720.24&99.28&3644.692&101159.7\\
		
		[5,5]&22322&98109&196098&17.955&149.728&1741562& 2595.8&148.85&5558.24&155350.8\\
		
		[6,6]&31614&139455&278766&24.824&209.2192&2480910&3651.36&208.42&7876.788&221219\\
		
		[7,7] &42526&188091&376014&32.793&279.2192&3351478&4886.92&277.99&10600.34&298764.4\\
		
		[8,8]&55058&244017&487842&41.862&358.9648&4353266&6302.48&357.56&13728.88&387987\\
		
		[9,9]&69210&307233&614250&52.031&448.7104&5486274&7898.48&447.13&17262.43&488886.8\\
		
		[10,10]&84982&377739&755238&63.3&548.456&6750502&9673.6&546.7&21200.98&601463.8\\
		\hline
	\end{tabular}
\end{table}
\begin{figure}[H]
	\centering
	\includegraphics[width=4in]{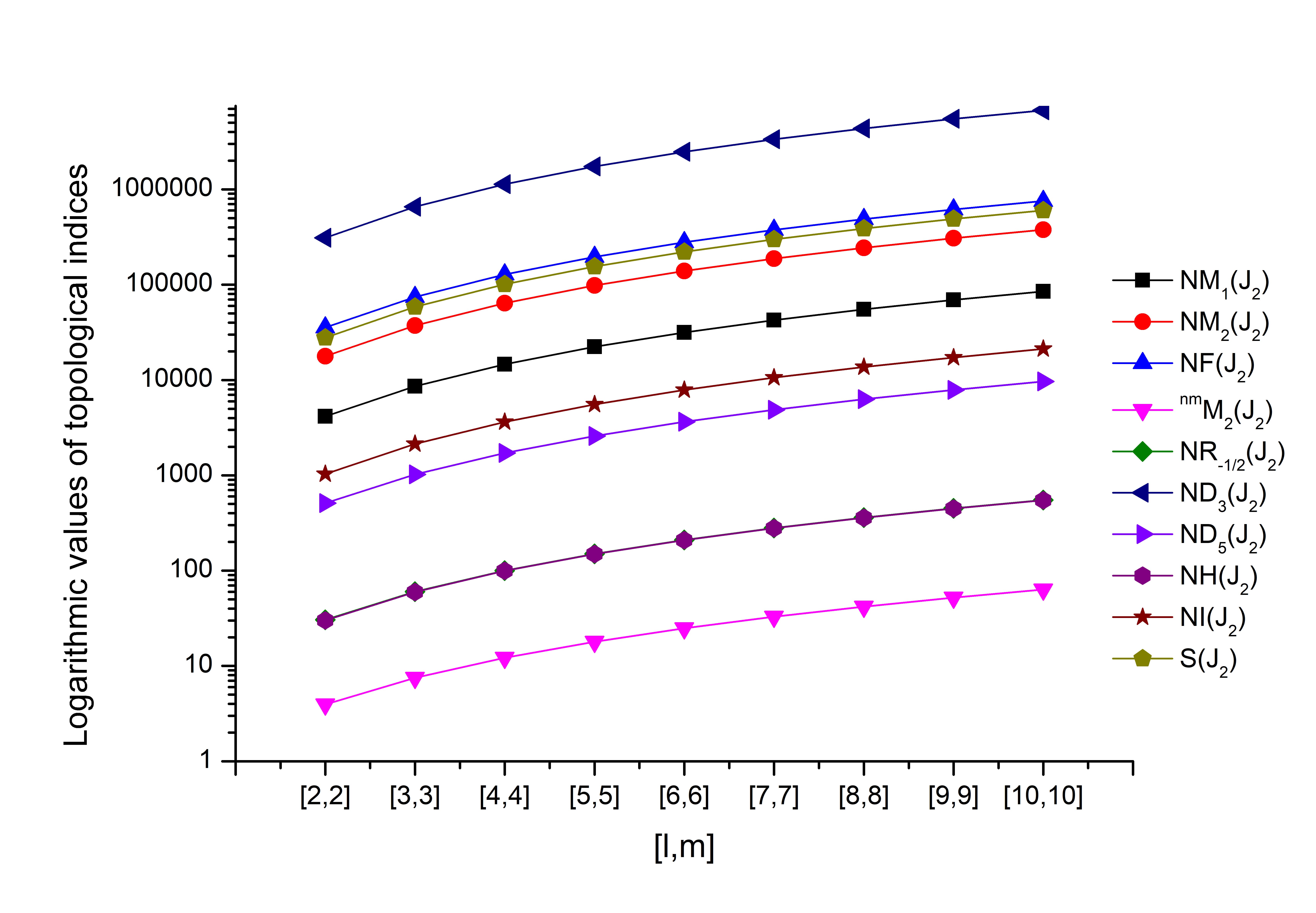}
	\caption{Graphical comparison among topological indices of Y-junction $J_2$ }
\end{figure}
Table $10$ shows some numerical values of topological indices of Y-junction $J_3$. Figure $5$ depicts the graphical comparison of these indices. Table $10$ and Figure $5$ show that the values of topological indices strictly increase as the values of $l$ and $m$ increases.\\
 From Tables $7$, $8$, $9$, and $10$, we see that as the values of $l$ and $m$ in Y-junction graphs increases, the corresponding values of topological indices grew very fastly. 
\begin{table}[H]
	\centering
	\caption{Numerical values of topological indices of Y-junction graph $J_3$}
	\setlength{\tabcolsep}{7pt}
	\renewcommand{\arraystretch}{1.8}
	\tiny\begin{tabular}{lccccccccccc}
		\hline 
		[$l,m$]& $NM_1(J_3)$&$NM_2(J_3)$&$NF(J_3)$&$N^mM_2(J_3)$&$NR_{-\frac{1}{2}}(J_3)$&$ND_3(J_3)$&$ND_5(J_3)$&$NH(J_3)$&$NI(J_3)$&$S(J_3)$\\
		\hline
		[2,2]&4302&18399&37278&4.15&31.6419&321774&529.8&31.18&1064.4&28694\\
		
		[3,3]&8802&38169&77058&7.82&61.9628&672930&1055.7&61.27&2183.85&59973\\
		
		[4,4]&14922&65229&131418&12.59&102.284&1155306&1761.6&101.36&3708.3&102929\\
		
		[5,5]&22662&99579&200358&18.46&152.605&1768902& 2647.5&151.45&5637.75&157562\\
		
		[6,6]&32022&141219&283878&25.43&212.926&2513718&3713.4&211.54&7972.2&223873\\
		
		[7,7] &43002&190149&381978&33.5&283.247&3389754&4959.3&281.63&10711.7&301861\\
		
		[8,8]&55602&246369&494658&42.67&363.568&4397010&6385.2&361.72&13856.1&391526\\
		
		[9,9]&69822&309879&621918&52.94&453.889&5535486&7991.1&451.81&17405.6&492868\\
		
		[10,10]&85662&380679&763758&64.31&554.209&6805182&9777&551.9&21360&605887\\
		\hline
	\end{tabular}
\end{table}
\begin{figure}[H]
	\centering
	\includegraphics[width=4in]{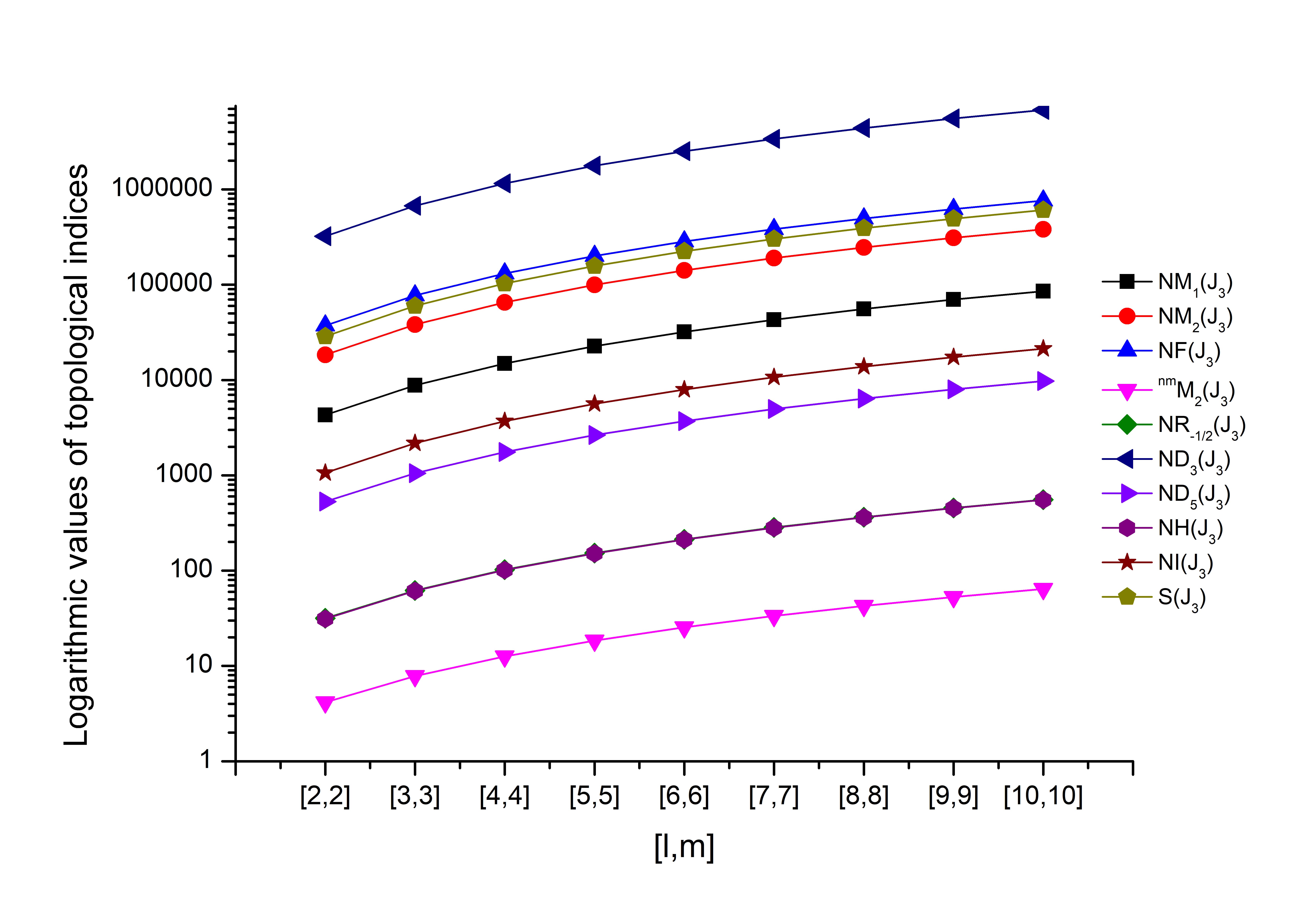}
	\caption{Graphical comparison among topological indices of Y-junction graph $J_3$}
\end{figure}
A few values of graph index-entropies of Y-junction graph $J$ are listed in Table $11$ and illustrated in Figure $6$. From Figure $6$, we see that entropy measures of $H_{\beta_1}$, $H_{\beta_2}$, $H_{\beta_3}$, and $H_{\beta_8}$ almost coincide. 
\begin{table}[h]
	\centering
	\caption{Numerical values of
		index-entropies of $J$}
	\setlength{\tabcolsep}{10pt}
	\renewcommand{\arraystretch}{1.8}
	\tiny\begin{tabular}{l|ccccccccc}
		\hline 
		[$l,m$]& $H_{\beta_1}(J)$&$H_{\beta_2}(J)$&$H_{\beta_3}(J)$&$H_{\beta_4}(J)$&$H_{\beta_5}(J)$&$H_{\beta_6}(J)$&$H_{\beta_7}(J)$&$H_{\beta_8}(J)$\\
		\hline
		[2,2]&2.363878&2.349537&2.351098&2.293045&2.469266&1.849911&2.235934&2.358964\\
		
		[3,3]&2.680031&2.668041&2.669162&2.614962&2.751708&2.162725&2.567487&2.674998\\
		
		[4,4]&2.912369&2.901799&2.902659&2.851695&2.966026&2.393078&2.813031&2.907277\\
		
		[5,5]&3.09586&3.086229&3.086919&3.038506&3.138274&2.575276& 3.00722&3.090734\\
		
		[6,6]&3.24743&3.238462&3.239031&3.192634&3.28218&2.725919&3.167437&3.242282\\
		
		[7,7] &3.37652&3.368045&3.368525&3.323745&3.40572&2.85433&3.303596&3.371357\\
		
		[8,8]&3.488993&3.480834&3.481247&3.437785&3.51397&2.966216&3.421864&3.483755\\
		
		[9.9]&3.588472&3.580678&3.581037&3.538669&3.610178&3.065343&3.526328&3.583289\\
		
		[10,10]&3.677788&3.670239&3.670554&3.629107&3.696846&3.154321&3.61983&3.672599\\
		\hline
	\end{tabular}
\end{table}
\begin{figure}[H]
\centering
\includegraphics[width=4in]{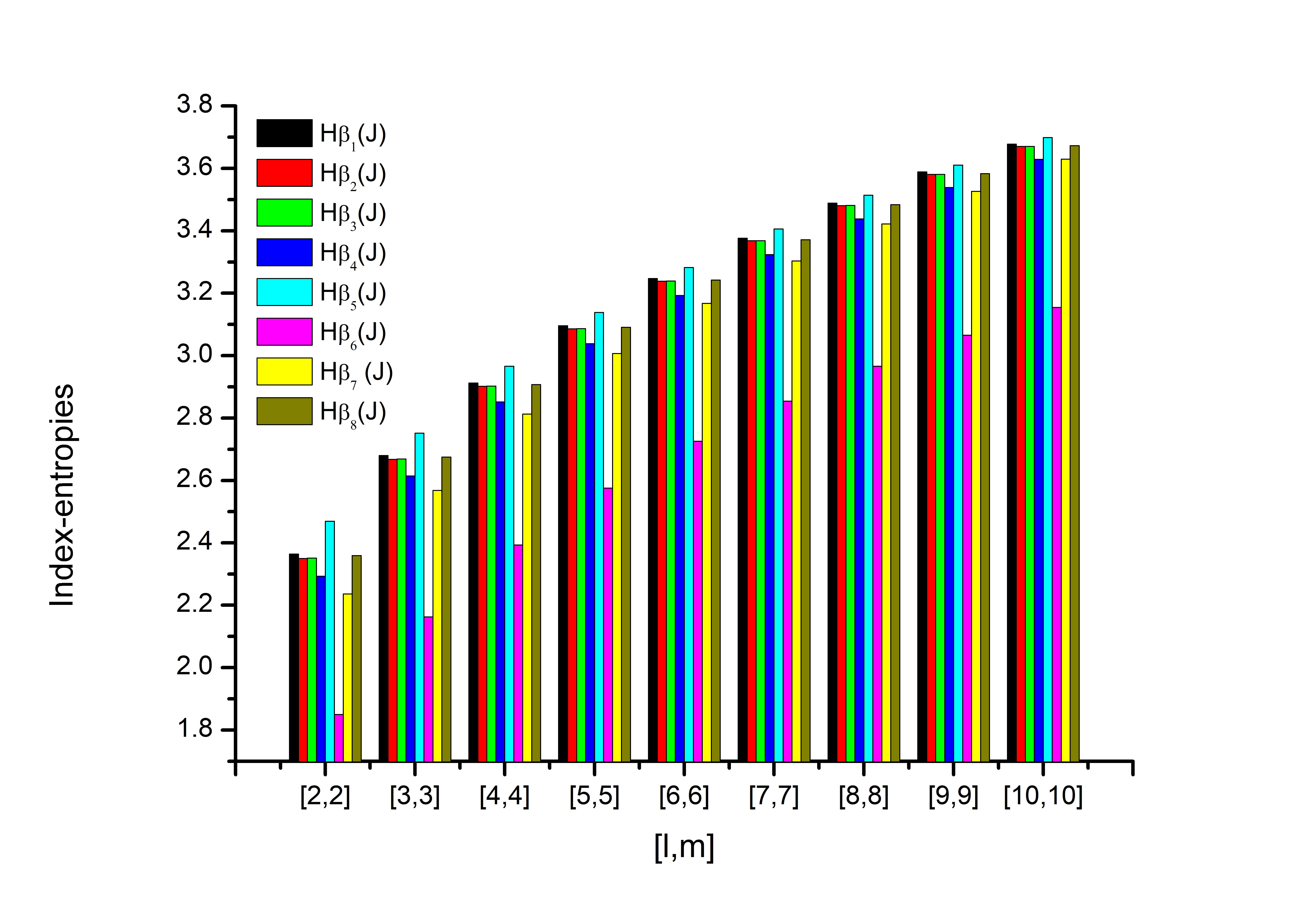}
\caption{Graphical comparison among index-entropies of J }
\end{figure}
The values of index-entropy of Y-junction graph $J_1$ is listed in Table $12$ and illustrated in Figure $7$. From Table $12$ and Figure $7$, we find that measures of  graph index-entropies $H_{\beta_1}$, $H_{\beta_2}$, $H_{\beta_3}$, $H_{\beta_5}$, $H_{\beta_6}$, and $H_{\beta_1}$ are almost same. 
\begin{table}[H]
\centering
\caption{Numerical values of
	index-entropies of $J_1$}
\setlength{\tabcolsep}{10pt}
\renewcommand{\arraystretch}{2}
\tiny\begin{tabular}{lccccccccc}
	\hline 
	[$l,m$]& $H_{\beta_1}(J_1)$&$H_{\beta_2}(J_1)$&$H_{\beta_3}(J_1)$&$H_{\beta_4}(J_1)$&$H_{\beta_5}(J_1)$&$H_{\beta_6}(J_1)$&$H_{\beta_7}(J_1)$&$H_{\beta_8}(J_1)$\\
	\hline
	[2,2]&2.374116&2.362875&2.365677&2.37483&2.351939&2.381171&2.336108&2.373399\\
	
	[3,3]&2.686117&2.677718&2.679751&2.705906&2.66972&2.691361&2.643717&2.686081\\
	
	[4,4]&2.916047&2.909359&2.91094&2.948613&2.903076&2.92019&2.871061&2.916411\\
	
	[5,5]&3.097981&3.092424&3.093709&3.139639&3.087253&3.101393& 3.051307&3.098605\\
	
	[6,6]&3.248463&3.243705&3.244784&3.2969&3.239312&3.251355&3.200605&3.24927\\
	
	[7,7] &3.376753&3.372588&3.373514&3.430435&3.368768&3.379258&3.23802&3.377694\\
	
	[8,8]&3.488549&3.484842&3.485651&3.546395&3.481462&3.490754&3.439143&3.489594\\
	
	[9.9]&3.587606&3.584263&3.584979&3.648847&3.58132&3.589572&3.537668&3.588733\\
	
	[10,10]&3.676529&3.673482&3.674123&3.740586&3.6707383&3.6783&3.626158&3.677722\\
	\hline
\end{tabular}
\end{table}
\begin{figure}[H]
	\centering
	\includegraphics[width=4in]{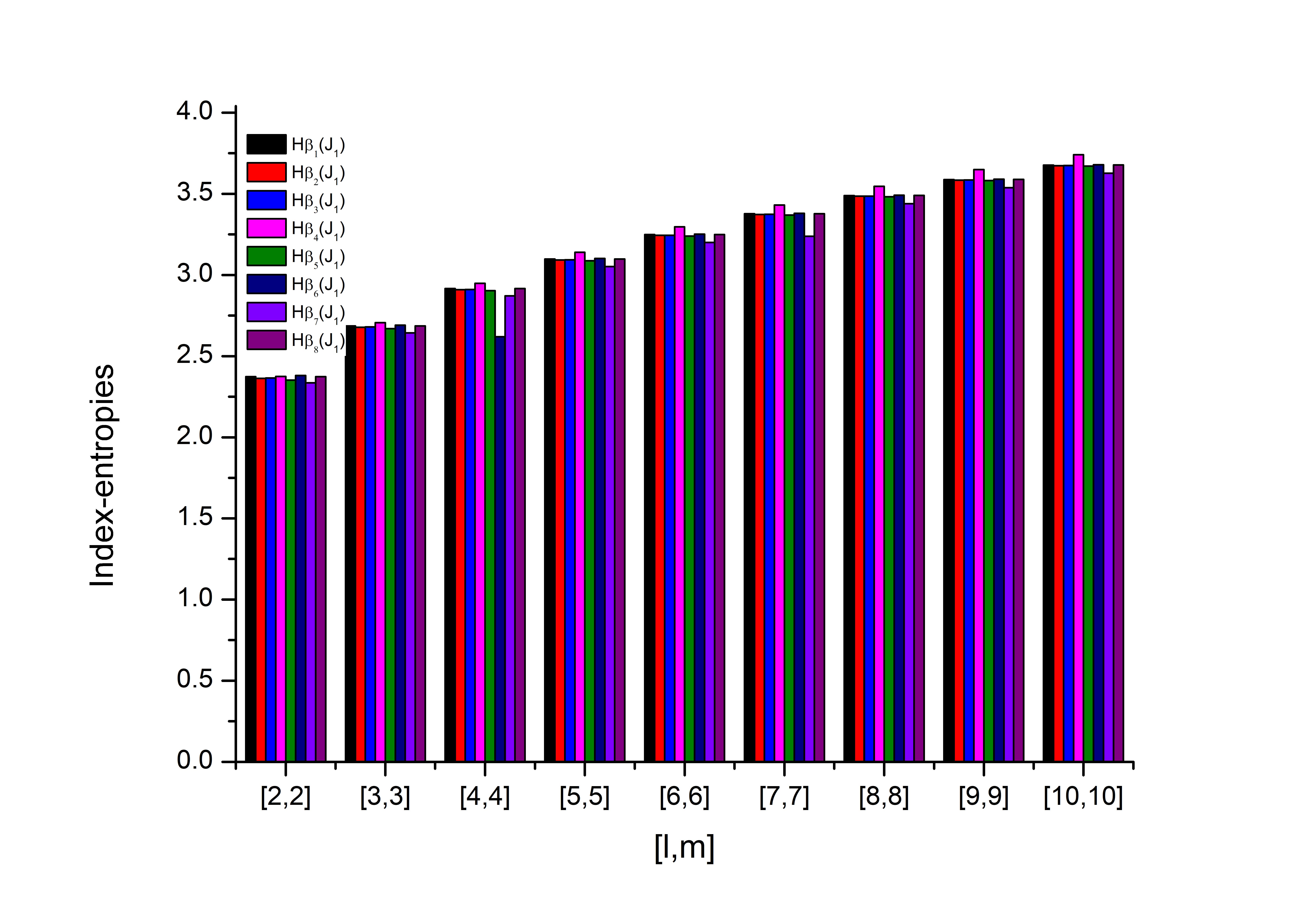}
	\caption{Graphical comparison among index-entropies of $J_1$ }
\end{figure}
Table $13$ depicts some graph index-entropies of Y-junction graph $J_2$. The graphical comparison of index-entropies of Y-junction graph $J_2$ is shown in Figure $8$. From Figure $8$, we see that graph index-entropies of $J_2$ increases as the values of $l$ and $m$ increases.
\begin{table}[H]
		\centering
		\caption{Numerical values of
			index-entropies of $J_2$}
		\setlength{\tabcolsep}{10pt}
		\renewcommand{\arraystretch}{1.8}
		\tiny\begin{tabular}{lccccccccc}
			\hline 
			[$l,m$]& $H_{\beta_1}(J_2)$&$H_{\beta_2}(J_2)$&$H_{\beta_3}(J_2)$&$H_{\beta_4}(J_2)$&$H_{\beta_5}(J_2)$&$H_{\beta_6}(J_2)$&$H_{\beta_7}(J_2)$&$H_{\beta_8}(J_2)$\\
			\hline
			[2,2]&2.388128&2.375827&2.346955&1.633105&2.336917&2.391354&2.345766&2.35432\\
			
			[3,3]&2.696411&2.687189&2.666479&1.88376&2.665889&2.698832&2.650775&2.672367\\
			
			[4,4]&2.924151&2.916793&2.9007&2.074455&2.902766&2.926068&2.876596&2.905747\\
			
			[5,5]&3.104655&3.098533&3.085384&2.229346&3.088295&3.106231& 3.055853&3.089892\\
			
			[6,6]&3.254134&3.248888&3.237774&2.360107&3.240916&3.255465&3.204459&3.241908\\
			
			[7,7] &3.381681&3.377087&3.367462&2.473394&3.370602&3.3882828&3.331363&3.371323\\
			
			[8,8]&3.492905&3.488816&3.480327&2.573399&3.48337&3.49391&3.442095&3.483978\\
			
			[9.9]&3.59151&3.587821&3.580028&2.662948&3.583139&3.592399&3.540309&3.583714\\
			
			[10,10]&3.680065&3.676703&3.659833&2.744042&3.672602&3.680861&3.628548&3.673185\\
			\hline
		\end{tabular}
	\end{table}
\begin{figure}[H]
	\centering
	\includegraphics[width=4in]{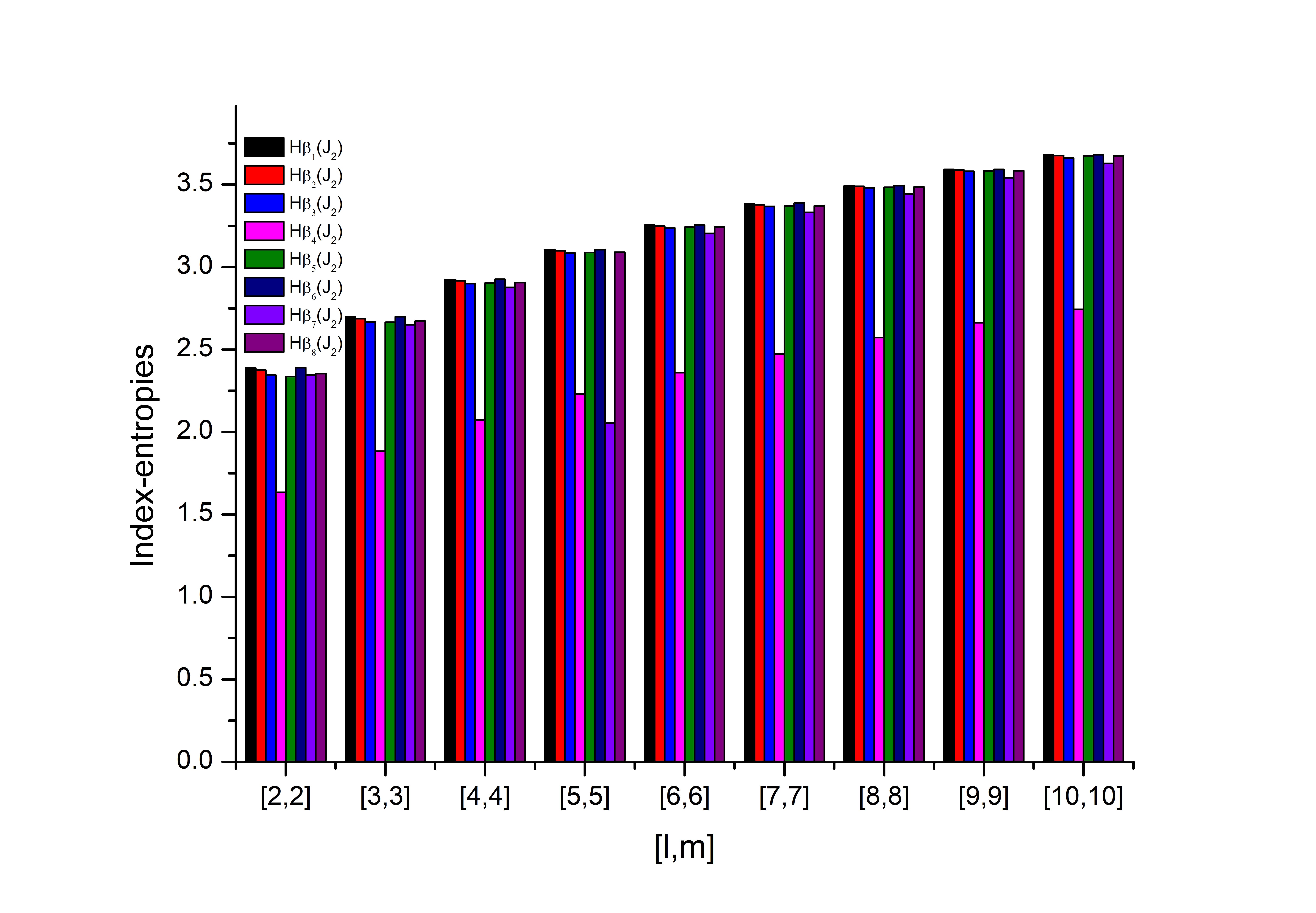}
	\caption{Graphical comparison among index-entropies of $J_2$ }
\end{figure}
In Table $14$, we calculate some graph index-entropies of Y-junction graph $J_3$. Figure $9$ shows the graphical comparison among index-entropies of $J_3$. From Table $14$ and Figure $9$, we see that index entropies $H_{\beta_1}$, $H_{\beta_2}$, $H_{\beta_3}$, $H_{\beta_6}$, and $H_{\beta_8}$ of $J_3$ are almost same. Also, Tables $11$, $12$, $13$, and $14$ shows that graph index-entropies of Y-junction graph increases as the values of $l$ and $m$ increases.
	\begin{table}[H]
		\centering
		\caption{Numerical values of
			index-entropies of $J_3$}
		\setlength{\tabcolsep}{10pt}
		\renewcommand{\arraystretch}{1.8}
		\tiny\begin{tabular}{lccccccccc}
			\hline 
			[$l,m$]& $H_{\beta_1}(J_3)$&$H_{\beta_2}(J_3)$&$H_{\beta_3}(J_3)$&$H_{\beta_4}(J_3)$&$H_{\beta_5}(J_3)$&$H_{\beta_6}(J_3)$&$H_{\beta_7}(J_3)$&$H_{\beta_8}(J_3)$\\
			\hline
			[2,2]&2.401692&2.388393&2.393488&2.179494&1.755856&2.404539&2.359174&2.40079\\
			
			[3,3]&2.706459&2.696449&2.7002&2.469933&2.242514&2.708584&2.660759&2.70624\\
			
			[4,4]&2.932101&2.924095&2.927043&2.687&2.571439&2.933776&2.884517&2.932298\\
			
			[5,5]&3.11224&3.104551&3.106972&2.86049&2.816575&3.112596& 3.062409&3.111698\\
			
			[6,6]&3.259727&3.254003&3.256051&3.005032&3.010781&3.260882&3.210048&3.260399\\
			
			[7,7] &3.38655&3.381534&3.383306&3.128925&3.171078&3.387542&3.336232&3.387369\\
			
			[8,8]&3.497215&3.492748&3.494307&3.237341&3.307302&3.498082&3.446407&3.498149\\
			
			[9.9]&3.595375&3.591345&3.592735&3.33372&3.425608&3.596141&3.544179&3.5964\\
			
			[10,10]&3.683569&3.679896&3.681147&3.420469&3.530087&3.684252&3.632058&3.684668\\
			\hline
		\end{tabular}
	\end{table}
\begin{figure}[H]
	\centering
	\includegraphics[width=4in]{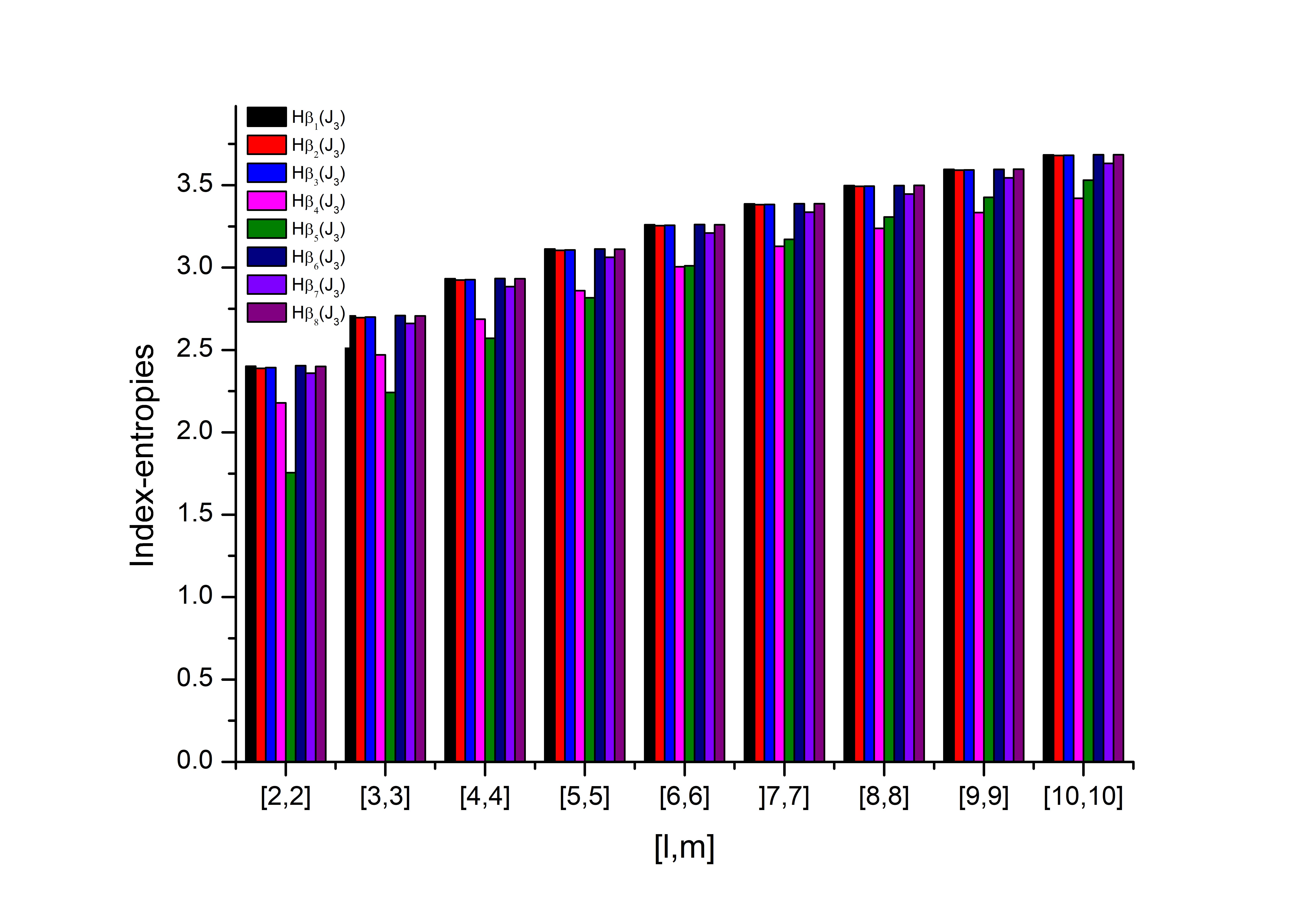}
	\caption{Graphical comparison among index-entropies of $J_3$ }
\end{figure}

\section{Conclusion and Future work} 
 In this study, the general expression of NM-polynomial for carbon nanotube Y-junction graphs is derived. Also, various neighborhood degree sum-based topological indices are retrieved from the expression of these polynomials. In addition, eight graph entropies in terms of these topological indices have been defined and calculated for Y-junction graphs. Furthermore, some numerical values of topological indices and index-entropies of Y-junction graphs are plotted for comparison. Since topological indices based on the degree of vertices has a significant ability to predict various physicochemical properties and biological activities of the chemical molecule. Therefore, the study's findings will be a viable option for predicting various physicochemical properties and understanding the structural problems of carbon nanotube Y-junctions.\\
 
 We mention some possible directions for future research, including multiplicative topological indices, graph index-entropies, regression models between the index-entropies and the topological indices, metric and edge metric dimension, etc., to predict thermochemical data, physicochemical properties, and structural information of carbon nanotube Y-junctions. \\

\noindent\textbf{Data Availability}\\
\noindent No data was used to support the findings of this study.\\

\noindent\textbf{Conflicts of Interest}\\ 
\noindent There are no conflicts of interest declared by the authors.\\

\noindent\textbf{Funding Statement}\\
\noindent The authors received no specific funding for this study.\\

\noindent\textbf{Author's Contribution Statement}\\
\noindent The final draft was written by \textbf{Sohan Lal} and \textbf{Vijay Kumar Bhat}. Figures and Tables are prepared by \textbf{Sohan Lal} and \textbf{Sahil Sharma}. All authors reviewed and edited the final draft.

\end{document}